\documentclass{article}

     \usepackage[preprint,nonatbib]{neurips_2020}

\usepackage{microtype}
\usepackage{graphicx}
\usepackage{subfigure}
\usepackage{booktabs} 
\usepackage{color}
\usepackage{physics}
\usepackage{amsthm}
\usepackage{comment}
\usepackage[utf8]{inputenc} 
\usepackage[T1]{fontenc}    
\usepackage{hyperref}       
\usepackage{url}            
\usepackage{booktabs}       
\usepackage{amsfonts}       
\usepackage{nicefrac}       
\usepackage{microtype}      
\usepackage{ amssymb }
\usepackage{algorithm}
\usepackage{algpseudocode}
\usepackage{ulem}

\title{A Robust Functional EM Algorithm for Incomplete Panel Count Data}

\author{%
  Alexander Moreno\\
  Georgia Institute of Technology
  \And
  Zhenke Wu\\
  University of Michigan\\
  \And
  Jamie Yap\\
  University of Michigan\\
  \And Cho Lam\\
  University of Utah\\
  \And David Wetter\\
  University of Utah\\
  \And Inbal Nahum-Shani\\
  University of Michigan\\
  \And Walter Dempsey\\
  University of Michigan\\
  \And James M. Rehg\\
  Georgia Institute of Technology
}

\begin{document}

\maketitle
\newtheorem{claim}{Claim}
\newtheorem{lemma}{Lemma}
\newtheorem{corollary}{Corollary}
\newtheorem{theorem}{Theorem}
\newtheorem{definition}{Definition}
\newtheorem{assumption}[]{Assumption}
\newtheorem{proposition}{Proposition}
\newcommand\numberthis{\addtocounter{equation}{1}\tag{\theequation}}
\newcommand{\cmmnt}[1]{}
\begin{abstract}
  Panel count data describes aggregated counts of recurrent events observed at discrete time points.  To understand dynamics of health behaviors and predict future negative events, the field of quantitative behavioral research has evolved to increasingly rely upon panel count data collected via multiple self reports, for example, about frequencies of smoking using in-the-moment surveys on mobile devices. However, missing reports are common and present a major barrier to downstream statistical learning. As a first step, under a missing completely at random assumption (MCAR), we propose a simple yet widely applicable functional EM algorithm to estimate the counting process mean function, which is of central interest to behavioral scientists. The proposed approach wraps several popular panel count inference methods, seamlessly deals with incomplete counts and is robust to misspecification of the Poisson process assumption.  Theoretical analysis of the proposed algorithm  provides finite-sample guarantees by expanding parametric EM theory \cite{balakrishnan2017statistical,wu2016convergence} to our general non-parametric setting. We illustrate the utility of the proposed algorithm through numerical experiments and an analysis of smoking cessation data. We also discuss useful extensions to address deviations from the MCAR assumption and covariate effects.
  \end{abstract}
\section{Introduction}
A major goal in behavioral medicine is identifying temporal patterns of risk factors preventing an individual from successfully modifying a health-related behavior. In smoking cessation, one would like to know times of day, locations, and other contextual factors such as smoking opportunity~\cite{kirchner2013geospatial} that may precipitate lapse to inform interventions to prevent lapse~\cite{rehg2017mobile}. A basic task is to describe when smoking occurs through modeling the \textit{counting process} of repeated negative events. One goal is to estimate the mean function of the counting process to characterize the dynamics of health-related behaviors at the population level.

The most common and widely-used measurement for behaviors is self-report via Ecological Momentary Assessment (EMA): a participant responds to prompts and enters data via a phone app~\cite{shiffman2008ecological}. An EMA can be randomly triggered by the app. Random EMAs provide a less-biased data collection paradigm relative to user initiated EMAs and have been used in a wide variety of behavioral studies~\cite{shiffman2007prediction,piasecki2011subjective,o1998coping}. Here, a participant is prompted 3-4 times a day at \textit{random} times. As EMA times will not generally co-occur with events of interest, participants report the \textit{cumulative number of events} since the last EMA. For example, how many cigarettes they smoked since the last prompt. This takes the form of~\textit{panel count data}~\cite{kalbfleisch1985analysis}\cite{sun1995estimation}, where exact times of a counting process are unobserved. Only the cumulative number of events since the last observation are measured, where observation times are assumed to follow some unknown distribution.


The counting process \textit{mean function} can describe population-level temporal patterns of smoking. It converts the discrete patterns of smoking counts from a population of participants (see Fig.~\ref{fig:raw-counts}) into a temporally-continuous summary of smoking behavior (see Fig.~\ref{fig:results}b). However, the \textit{missingness} inherent in EMA data makes  consistent estimation of the mean function difficult. There are several conditions causing missingness. An EMA may be ignored by the user or opened and then abandoned. Second, the mobile app itself may postpone the triggering of an EMA for any of several reasons (e.g. battery low). While EMAs are triggered randomly, the random process is constrained to have a minimum temporal spacing between EMAs, in order to keep participant burden at an acceptable level. If EMAs are postponed or ignored too many times, it will not be possible to trigger the full set of EMAs for the day, resulting in missing EMAs.  Missing data is a common issue in studies that involve EMAs, with \cite{jones2019compliance} noting that over 126 studies, the average missingness rate is 25\%. 

For mean function estimation, missing EMAs cause problems when they lead to inaccurate counts of the total cigarettes smoked between EMAs. Due to recall bias, longer intervals between EMAs are less reliable.\footnote{In principle a participant who completes very few EMAs but reports their counts with 100\% accuracy is not creating missing data, because standard mean function estimators~\cite{wellner2000two} would still be asymmpotically unbiased. In practice, however, behavioral scientists frequently treat EMAs as if they are missing if the interval to the last EMA exceeds a cut-off, such as 24 hours, due to recall bias.} Behavioral scientists have developed heuristic imputation schemes to adjust for missing counts~\cite{hallgren2013missing}, but these may not consistently estimate the mean function. \textit{This paper presents the first self-contained and systematic treatment of missing data for panel count data, by providing a simple Expectation-Maximization (EM)  algorithm \cite{dempster1977maximum,balakrishnan2017statistical} for estimating the mean function with finite-sample theoretical guarantees.}

Our primary methodological contribution is a functional EM algorithm to wrap standard non-parametric mean function estimators to handle missing data under a missing completely at random (MCAR) assumption \cite{little2019statistical}. The E-step uses estimates from a fitting method to impute missing data, and the M-step calls that fitting method to estimate a mean function. This extends several classic non-parametric methods \cite{wellner2000two,lu2007estimation} and the baseline-only version of a semi-parametric mean function estimation method \cite{wang2013augmented} to  the setting of missing data.

We analyze our EM algorithm using the frameworks described in \cite{balakrishnan2017statistical,wu2016convergence} and obtain finite sample guarantees. This requires care as we are extending their work from the parametric to the non-parametric setting. This paper addresses three major theoretical challenges in  the context of functional EM by: (i) noting that an infinite dimensional derivative in our setting is analogous to the inner products used in \cite{wu2016convergence,balakrishnan2017statistical}, (ii) showing local uniform strong concavity of our population E-step, and (iii) a high probability finite sample bound on the convergence of the M-step. The lack of ground truth for missing data in real-world settings makes evaluation difficult and theoretical guarantees important. Finally, our proposed algorithm can consistently estimate the mean function even when the Poisson process assumption is violated. This is achieved by recovering the population MLE, and noting that under certain integrability conditions the population MLE of the Poisson process log-likelihood is the true mean function of the counting process, even when the counting process is not Poisson.


\section{Related Work}

\textbf{Ecological Momentary Assessments} are frequently used to record counts for behavior, including smoking counts \cite{shiffman2009many,wang2012truth,griffith2009method}, alcohol counts \cite{collins2003feasibility} and promiscuous behaviors \cite{wray2016using}.  However, none of these papers use panel count data methods to estimate the mean function.

\textbf{Panel count data} analysis comes from nonparametric statistics, but our missing data problem has not been addressed. The closest works to ours are \cite{wang2013augmented,wellner2000two,lu2007estimation}. These papers estimate the mean function from panel count data, but cannot account for missing EMAs as in our motivating application. \cite{wang2013augmented} develop an EM approach under a more limited missingness model assuming that the total counts for each participant are known prior to data analysis. This is unsuitable for analyzing data with missing EMAs. However, two strengths are that they can incorporate baseline covariate information, and they relax the standard assumption of independence between observation times and the counting process.

Neither of
\cite{wellner2000two,lu2007estimation} handle missing counts, and cannot be applied directly to our setting. We propose an EM algorithm for missing data to wrap any of \cite{wang2013augmented,wellner2000two,lu2007estimation}, building on existing methods in the M-step.  \cite{wellner2000two} provides the first strong theoretical guarantees for a mean function estimator in the panel count setting.  \cite{lu2007estimation} improve their rate of convergence with spline mean functions.

Many other papers in the literature focus on either extending to the semi-parametric setting, relaxing assumptions in either the non-parametric or semi-parametric setting, or improving estimation under specific distributional assumptions.  They do not address missing data.  \cite{wellner2007two} extended \cite{wellner2000two} to the semi-parametric setting, \cite{huang2006analysing} analyzed panel count with informative observation times and subject-specific frailty, and \cite{hua2014spline} proposed using a smooth semi-parametric estimator to handle over-dispersion in panel count data.  \cite{ding2018variational} proposed using a squared Gaussian process intensity function.  

\textbf{EM Theory.} Also relevant is recent work on convergence of the EM algorithm to the true parameter \cite{balakrishnan2017statistical,wu2016convergence}.  In \cite{balakrishnan2017statistical}, they show that under good initialization and certain regularity conditions, population EM converges to the true parameter.  With finite sample uniform convergence of the M-step, one can also show a finite-sample bound for EM.  \cite{wu2016convergence} show that by assuming the ability to optimize over a ball around the true parameter, one can weaken regularity conditions and replace uniform convergence of the $M$-step with three concentration inequalities, often easier to prove.  We base our population guarantees on \cite{wu2016convergence}.  We base our finite-sample guarantees on \cite{balakrishnan2017statistical}, because \cite{wu2016convergence} assumes that empirical and population norms are the same which does not hold in our non-parametric setting.

\begin{figure}
    \centering
    \includegraphics[width=0.5\textwidth]{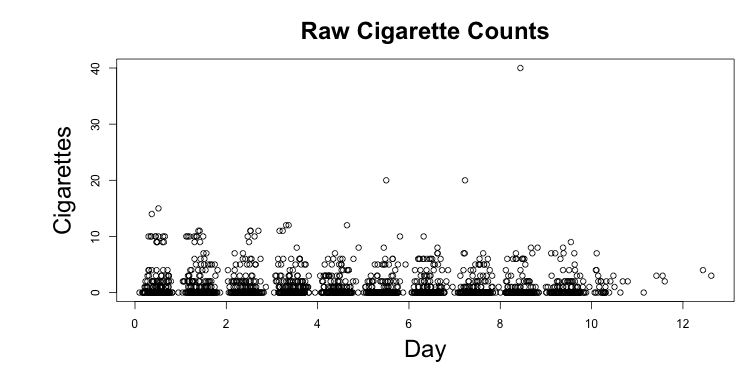}
    \caption{Raw cigarette counts smoked between observations.  We want to convert this to a mean function of cumulative cigarettes smoked (Figure \ref{fig:results}b), but some counts are missing or unreliable.}
    \label{fig:raw-counts}
\end{figure}
\section{Model}
\subsection{Complete Data}
Let $N=\{N(u):u\geq 0\}$ be a univariate counting process.  The goal is to estimate the mean function $\Lambda^*(u)=\mathbb{E}[N(u)]$ over a study window $[0,\tau]$.  Let $K$ be the random number of observations for a participant.  Let $T=(T_1,\cdots,T_K)\in \mathbb{R}_+^K$ be a random vector of observation times, with $T_0=0$.  Let $\Delta N=(\Delta N_1,\cdots,\Delta N_K)$ be the count increments: $\Delta N_j=N(T_j)-N(T_{j-1})$. Let $\Delta \Lambda(T_j)=\Lambda(T_j)-\Lambda(T_{j-1})$ be mean function increments. For each participant $i=1,\cdots,n$, we observe $Y=(\Delta N,T,K)\in \mathcal{N}\times \mathcal{T}\times \mathcal{K}$ where $\mathcal{N},\mathcal{T},\mathcal{K}$ are the corresponding sample spaces.    Let $P_n,P$ denote the empirical and true measures on $\mathcal{N}\times \mathcal{T}\times \mathcal{K}$, respectively.  For a measurable function $f$, let $P_n f=\frac{1}{n}\sum_{i=1}^n f(Y_i)$ and $Pf=\int fdP$.  Under a non-homogeneous Poisson process, the complete data sample log-likelihood is
\begin{align*}
    l_n(Y|\Lambda)&=nP_n\left\{ \sum_{j=1}^{K}\Delta N_{j}\log[\Delta \Lambda(T_j)]-\Lambda(T_{K})\right\},
\end{align*}
and the population log-likelihood is similar but with $P$ instead of $nP_n$. The goal is consistent mean function estimation even when the Poisson process assumption is violated.

\subsection{Missing Data}

Unlike previous work, we assume certain observations $\Delta N_j$ are missing.  For $\Delta N\in \mathbb{R}^K$, each observation $\Delta N_j$ is missing completely at random with probability (w.p.) $\epsilon\in [0,1)$.  Let $\tau\in \{\circ,1\}^K$ be the missingness pattern, where $\circ=0$ but the notation represents missingness.  Let $s=1-\tau$, thus $s_j=1$ if $\Delta N_j$ is missing.  Let $\Delta N^{(\tau)}= \Delta N\odot\tau$ and $\Delta N^{(s)}= \Delta N\odot s$, where $\odot$ is elementwise product.  Then $\Delta N=\Delta N^{(s)}+\Delta N^{(\tau)}$.  We observe $\Delta N^{(\tau)}$ but not $\Delta N^{(s)}$.  Let $Z\equiv\Delta N^{(s)}$ represent our missing data vector, and let $\mathcal{Z}$ be the space of values for our missing data.  Then $Y=(\Delta N^{(\tau)},T,K)$, and $(Y,Z)\in \mathcal{N}\times \mathcal{T}\times \mathcal{K}\times \mathcal{Z}$ gives us the observed and missing parts of the data, respectively.  In this case $P_n$ and $P$ are now the empirical and true measures for $\mathcal{N}\times \mathcal{T}\times \mathcal{K}\times \mathcal{Z}$.

\section{EM Algorithm}
We first describe the general setting for EM. Let $f_{\Lambda}(y,z)$, $p_\Lambda(y)$, and $k_\Lambda(y|z)$ be joint, marginal, and conditional densities respectively. Let $\Theta$ be a convex set of functions and $\{\Theta_n\}$ a sieve: a nested sequence $\Theta_1\subset \Theta_2\subset \cdots $ such that $\cup_{n=1}^\infty \Theta_n\subseteq \Theta$ is dense in $\Theta$. We define the following.
\begin{definition}
(Population $Q$-function) $Q(\Lambda'|\Lambda)\equiv\int_{\mathcal{Y}}\left(\int_{\mathcal{Z}(y)}\log(f_{
    \Lambda'}(y,z))k_{\Lambda}(z|y)dz\right)p_{\Lambda^*}(y)dy$.
\end{definition}
\begin{definition}
(Sample $Q$-function) $Q_n(\Lambda'|\Lambda;\{Y_i\}_{i=1}^n)\equiv\frac{1}{n}\sum_{i=1}^n \mathbb{E}_{\Lambda}[\log f_{\Lambda'}(y,z)|Y_i]$.
\end{definition}
\begin{definition}
$B_r(\Lambda^*)\equiv \{\Lambda:\Vert \Lambda-\Lambda^* \Vert_\infty \leq r\}$ where $\Vert \cdot\Vert_\infty$ is the essential supremum.
\end{definition}
\begin{definition}
(Population $M$-step) $M(\Lambda^{(t)})=\arg\max_{\Lambda'\in \Theta\cap B_r(\Lambda^*)}Q(\Lambda'|\Lambda^{(t)})$.
\end{definition}
\begin{definition}
(Sample $M$-step) $M_n(\Lambda^{(t)})=\arg\max_{\Lambda'\in \Theta_n\cap B_r(\Lambda^*)}Q_n(\Lambda'|\Lambda^{(t)})$.
\end{definition}
A $Q$-function is an $E$-step of the EM algorithm.. Population EM repeatedly takes $\Lambda^{(t+1)}=M(\Lambda^{(t)})$, where $\Lambda^{(t)}$ denotes iteration $t$'s estimate. Sample EM repeatedly takes $\Lambda^{(t+1)}=M_n(\Lambda^{(t)})$. Next we describe EM computation.   Algorithm 1 describes sample EM; population EM is similar.

\begin{algorithm}[H]\label{alg:em-algorithm}
\begin{algorithmic}[1]
\State Initialize $\Lambda_0\in \Theta_n\cap B_r(\Lambda^*)$, let $t\leftarrow0$\;
\While{not converged}
\State(E-step): Compute $Q_n(\Lambda'|\Lambda^{(t)};\{Y_i\}_{i=1}^n)$ using current mean function estimate
\State(M-step): $\Lambda^{(t+1)}\leftarrow\arg\max_{\Lambda'\in \Theta_n\cap B_r(\Lambda^*)}Q_n(\Lambda'|\Lambda^{(t)};\{Y_i\}_{i=1}^n)$ using existing method
\State $t\leftarrow t +1$
\EndWhile
\end{algorithmic}
%
\caption{Sample-based EM Algorithm for Panel Count Data with Missing Counts.}
\end{algorithm}
Algorithm 1 assumes $B_r(\Lambda^\star)$ is known following \cite{wu2016convergence}: in practice it is unknown, and we suggest trying multiple initializations and choosing the one that leads to the highest final log-likelihood. In our numerical experiments and data analysis, for illustration we assume $\Theta_n = \{\textrm{monotone step functions with at most $n$ steps}\}$. Next we describe the E- and M-steps.

\subsection{E-Step}
The population and sample Q-functions replace missing counts with mean function estimates. The sample Q-function (the population version is similar but uses $P$ instead of $P_n$.) is
\begin{align*}
    Q_n(\Lambda'|\Lambda;\{Y_i\}_{i=1}^n)
    &=P_n\left\{\sum_{j=1}^K \Delta N_{j}^{\tau_j}\Delta \Lambda^{s_j}(T_j)\log[ \Delta\Lambda'(T_j)]-\Lambda'(T_{K})\right\}
\end{align*}
\subsection{M-Step}
The population M-step maximizes $Q(\Lambda'|\Lambda)$ using 
\[\Theta=\{\Lambda:[0,\tau]\rightarrow[0,\infty) \vert \Lambda\textrm{ is nondecreasing, }\Lambda(0)=0,\Lambda(\tau)\leq U_{\textrm{all}}\},\]  where $U_{\textrm{all}}$ is a uniform upper bound for functions in this set.  The sample $M$-step uses a convex set $\Theta_n\subset \Theta$ (sieve estimator).  $\Theta_n$ is potentially monotonic step functions \cite{wellner2000two,wang2013augmented} or monotonic splines \cite{lu2007estimation}, subject to the constraint of having the upper bound of $U_{\textrm{all}}$ over the study time.  Maximization proceeds via existing methods \cite{wellner2000two,lu2007estimation,wang2013augmented}. Like $B_r(\Lambda^*)$, $U_{\textrm{all}}$ is unknown.

\section{Theory}
Section \ref{sec:assumptions} defines assumptions, \ref{sec:measures-metrics-inner} defines distances and a quasi-inner product, and \ref{sec:pop-theory} gives two regularity condition lemmas based on \cite{wu2016convergence}.  Proposition \ref{prop:population-contraction} shows that with good initialization and sufficiently small missingness probability, a population EM step gives a contraction, moving our estimate closer to the true mean function after every iteration.  Theorem \ref{thm:population-convergence} then shows that population EM converges linearly to the true mean function.

We then show finite-sample theory. Proposition \ref{prop:rate-of-convergence} states that with high probability, the sample M-step converges uniformly to the population M-step. Theorem \ref{thm:sample-convergence} states that with high probability the distance between the current estimate and the true mean function is bounded by two terms: one describes applying population EM, and the other involves the uniform convergence of the $M$-step.

\subsection{Assumptions}\label{sec:assumptions}
We make the following assumptions, similar assumptions were made in \cite{wellner2000two,lu2007estimation,wang2013augmented}:
\begin{enumerate}
    \item The counting process $N$ is independent of the number of observations $K$ and observation times $T$, respectively;\label{assm:counting-obs-process-ind}
    \item The observation times are random variables taking values in the bounded set $[\tau_0,\tau]$ where $0<\tau_0<\tau$ and $\tau\in (0,\infty)$;\label{assm:random-observation-times}
    \item The number of observations is bounded, i.e., there exists $k_0>0$ such that $P(K\leq k_0)=1$;\label{assm:bounded-number-observations}
    \item The true mean function is uniformly bounded over the study, satisfying $\Lambda^*(u)\leq U\leq U_{\textrm{all}}$ for some $U\in (0,\infty)$ and all $u\in [\tau_0,\tau]$. Recall $U_{\textrm{all}}$ is a uniform upper bound on functions in $\Theta$; \label{assm:bounded-mean-function}
    \item The first derivative of $\Lambda^*(u)$ has a positive lower bound in $[\tau_0,\tau]$.\footnote{This is not strong, and doesn't imply at least a constant first derivative/superlinear function in general.  Consider $\Lambda^*(t)=t^{1/2}$ over $[0,\tau]$.  The derivative $\frac{1}{2\sqrt{t}}\geq \frac{1}{2\sqrt{\tau}}$, but this function is sub-linear.}\label{assm:bounded-derivatives}
    \item The observation times are $\alpha$-separated.  That is, $P(T_{j}-T_{j-1}\geq \alpha)=1$ for some $\alpha>0$ and all $j=1,\cdots,K$;\label{assm:alpha-separated}
    \item $\mathbb{E}[\exp(aN(t))]$ is uniformly bounded for $t\in [0,\tau]$ some constant $a$;\label{assm:unif-bounded-mgf}
    \item The count increments $\Delta N_{j}$ are missing completely at random (MCAR) w.p. $\epsilon>0$. That is, $s_{j_i},j_i=1,\cdots,K_i,i=1,\cdots,n$ are iid $\textrm{Bernoulli}(\epsilon)$ random variables.\label{assm:MCAR} 
\end{enumerate}
In the context of an observational smoking study, some of these assumptions can be scientifically expressed: \ref{assm:counting-obs-process-ind}) and \ref{assm:random-observation-times}) EMAs are delivered at random times throughout the study and are independent of smoking times \ref{assm:bounded-number-observations}) there is a maximum number of EMAs delivered over the study \ref{assm:bounded-derivatives}) there is a minimum smoking risk throughout the study \ref{assm:alpha-separated}) there is a minimum time between EMAs \ref{assm:MCAR}) the probability that an EMA is missed is $\epsilon>0$. The only new assumption for panel count is assumption \ref{assm:MCAR}.

\subsection{Measures, Distance Metrics, and Quasi-Inner Product}\label{sec:measures-metrics-inner}
We next define measures for constructing distance metrics between mean functions.  We then define a quasi-inner product that is used to show regularity conditions for population EM theory. Let $B,B_1,B_2$ be the intersection of Borel sets in $\mathbb{R}$ with $[0,\tau]$.
\begin{definition}[Measures for sets containing observation times]
$\mu(B)\equiv \mathbb{E}\left[\sum_{j=1}^K 1_B(T_{j})\right]$ and $\mu_2(B_1\times B_2)\equiv\mathbb{E}\left[\sum_{j=1}^K 1_{B_1}(T_{j-1})1_{B_2}(T_{j})\right].$
\end{definition}
  Then $\mu(B)=\mathbb{E}\vert \{\textrm{observations in }B\}\vert$ and $\mu_2(B_1\times B_2)=\mathbb{E}|\{\textrm{one observation in $B_1$, next in $B_2$}\}|$.
\begin{definition}[$d_2$ metric for mean functions]
$\Vert \Lambda_1-\Lambda_2\Vert\equiv[\int_0^\tau\int_0^\tau |(\Lambda_1(v)-\Lambda_1(u))-(\Lambda_2(v)-\Lambda_2(u))|^2d\mu_2(u,v)]^{1/2}.$
\end{definition}
This is the $d_2$ metric of \cite{wellner2000two}.  Under assumption \ref{assm:bounded-number-observations}, convergence in $\Vert \cdot\Vert$ implies convergence in $L^2(\mu)$.  We base our theory on convergence in this norm.  Now define
\begin{align*}
    \langle \nabla Q(\Lambda^{(l)}|\Lambda),\Lambda'\rangle
    &\equiv\lim_{\eta\downarrow 0}\frac{Q(\Lambda^{(l)}+\eta \Lambda'|\Lambda)-Q(\Lambda^{(l)}|\Lambda)}{\eta}\\
    &=P\left\{\sum_{j=1}^K (\frac{\Delta N(T_{j})^{\tau_j}\Delta\Lambda(T_{j})^{s_j}}{\Delta \Lambda^l_{j}}-1)(\Delta\Lambda'_{j})\right\}.\numberthis\label{eqn:inner-product}
\end{align*}

We do not claim that~\eqref{eqn:inner-product} is a valid inner product, but it is closely related to the inner product from \cite{wu2016convergence,balakrishnan2017statistical}, the $Q$-function's directional derivative in the direction of a parameter vector, while~\eqref{eqn:inner-product} is the $Q$-function's right Gateaux derivative in the direction of a mean function, analogous in function space. The connection to inner products from \cite{balakrishnan2017statistical,wu2016convergence} is key to using existing EM theory to prove our Lemmas \ref{lemma:gradient-stability} and \ref{lemma:almost-strong-concavity}. We prove Equation (\ref{eqn:inner-product}) in \ref{proof:inner-product} in the supplementary material.

\subsection{Population Theory}\label{sec:pop-theory}
Before stating our main population EM convergence theorem, we define important constants and state two lemmas for the population $Q$-function.  Lemmas \ref{lemma:gradient-stability} and \ref{lemma:almost-strong-concavity} mirror gradient stability and local uniform strong concavity from \cite{wu2016convergence}, respectively, but they studied $Q$-functions with quadratic dependence on parameters. Our objective function does not have quadratic dependence on the mean function. However, we can decompose our $Q$-function into a sum of two terms: one locally uniformly strongly concave, and one strictly concave. The sum is then locally uniformly strongly concave. Lemma \ref{lemma:almost-strong-concavity} shows this lower bound holds, and is a key step that allows us to apply \cite{wu2016convergence} to our setting. We then prove Proposition~\ref{prop:population-contraction} which states that population EM steps contract, bringing estimates closer to the true mean function after each iteration.

\begin{definition}
Let $c\equiv \inf \Delta \Lambda^*>0$ be a uniform lower bound on increments of the true mean function.  This exists by assumptions \ref{assm:bounded-derivatives} and \ref{assm:alpha-separated}.
\end{definition}
\begin{definition}
Let $b\equiv \sup \{\Delta \Lambda:\Lambda\in B_r(\Lambda^*)\}>0$ be a uniform upper bound on increments of mean functions in $B_r(\Lambda^*)$.  Note $b\leq U_{\textrm{all}}$, where $U_{\textrm{all}}$ is a uniform upper bound on functions in $\Theta$. 
\end{definition}
\begin{definition}
Let $\gamma=\frac{\epsilon}{c}$ and $\nu=\frac{1-\epsilon}{3b}$
\end{definition}
\begin{lemma}\label{lemma:gradient-stability}
(Gradient Stability) Assume assumptions \ref{assm:counting-obs-process-ind}, \ref{assm:random-observation-times}, \ref{assm:bounded-number-observations}, \ref{assm:bounded-mean-function}, \ref{assm:bounded-derivatives}, \ref{assm:alpha-separated}, and \ref{assm:MCAR} hold.  Then for  $\Lambda,\Lambda'\in \Theta$ (proof in \ref{proof:gradient-stability}),
\begin{align*}
    \langle \nabla Q(\Lambda^*|\Lambda)-\nabla Q(\Lambda^*|\Lambda^*),\Lambda'-\Lambda^*\rangle
    &\leq  \gamma \Vert \Lambda-\Lambda^*\Vert \Vert \Lambda'-\Lambda^*\Vert.
\end{align*}
\end{lemma}
\begin{lemma}\label{lemma:almost-strong-concavity}
(Local Uniform Strong Concavity) Assume all assumptions hold. Then if $r\leq\frac{c}{4}$ and $\Lambda',\Lambda\in B_r(\Lambda^*)$ (proof in \ref{proof:almost-strong-concavity}),
\begin{align*}
    Q(\Lambda^*|\Lambda)-Q(\Lambda'|\Lambda)+\langle \nabla Q(\Lambda^*|\Lambda),\Lambda'-\Lambda^*\rangle\geq \nu\Vert \Lambda'-\Lambda^*\Vert^2.
\end{align*}
\end{lemma}
One can think of this similarly to a second order Taylor expansion, where we're taking the function of $\Lambda'$ expanded at $\Lambda^*$ conditional on $\Lambda$, and using that to derive an inequality.
\begin{proposition}\label{prop:population-contraction}
(Population EM Contraction) Assume all assumptions hold. If $r\leq \frac{c}{4}$, $\Lambda',\Lambda\in B_r(\Lambda^*)$ and $Q(\Lambda'|\Lambda)\geq Q(\Lambda^*|\Lambda)$, then
\begin{align*}
    \Vert \Lambda'-\Lambda^*\Vert &\leq \frac{\gamma}{\nu}\Vert \Lambda-\Lambda^*\Vert.
\end{align*}
\end{proposition}
See section \ref{proof:population-contractivity} in the supplementary material for a detailed proof. This is similar to Proposition 3.2 in \cite{wu2016convergence}. In order for this to give a contraction, we need $\gamma<\nu$, which holds if $\epsilon<\frac{c}{3b+c}$. Thus if the uniform lower bound on increments of the true mean function goes up or the uniform upper bound on increments in the ball $B_r(\Lambda^*)$ goes down, we can tolerate a higher probability of missing data.

\begin{theorem}\label{thm:population-convergence}
(Population EM Convergence to True Mean Function) Suppose the assumptions of the above hold and $0<\gamma<\nu$.  Take the EM sequence $\Lambda^{(t+1)}=\arg\max_{\Lambda'\in \Theta\cap B_r(\Lambda^*)}Q(\Lambda'|\Lambda^{(t)})$ and $\Lambda^0\in B_r(\Lambda^*)\cap \Theta$.  Then
\begin{align*}
    \Vert \Lambda^{(t)}-\Lambda^*\Vert &\leq \left(\frac{\gamma}{\nu}\right)^t\Vert \Lambda^{(0)}-\Lambda^*\Vert.
\end{align*}
\end{theorem}

The proof is in \ref{proof:population-convergence} in the supplementary material.  At each EM step, we move towards the true function by a multiplicative factor of $\left(\frac{\gamma}{\nu}\right)$.  This is similar to Theorem 3.1 in \cite{wu2016convergence}.

\subsection{Sample Theory}
We next discuss convergence of the sample EM algorithm.  We again require the initial estimate $\Lambda^0$ and subsequent estimates $\Lambda^{(t)}$ to be close, i.e., in the set $B_r(\Lambda^*)\cap \Theta_n$.  The key additional assumption is that the sample size is large enough to satisfy certain conditions. For all $n$ greater than this minimum sample size, we can guarantee approximate contraction of the sample EM algorithm, where the term that does not contract goes to $0$ in large samples.

\begin{proposition}\label{prop:rate-of-convergence}
Suppose all assumptions hold. Assume $B_r(\Lambda^*)$ has radius $r\leq\frac{c}{4}$. For any $L>0$, there exists $u_L\overset{L\rightarrow\infty}{\rightarrow} 0$ such that w.p. $1-u_L$
\begin{align*}
    \Vert M_n(\Lambda)-M(\Lambda)\Vert\leq 2^L n^{-1/3}.
\end{align*}
\end{proposition}
See \ref{proof:rate-of-convergence} of the supplementary material for proof, as well as a definition for $u_L$. This is the rate of convergence an M-estimator, and generalizes the rate of convergence of \cite{balakrishnan2011class} from maximum likelihood estimates to M-steps. We now state our main result.

\begin{theorem}\label{thm:sample-convergence}
Suppose $0<r\leq \frac{c}{4}$ and $0<\gamma<\nu$ and all assumptions hold such that the population contractivity holds. Let $\kappa=\frac{\gamma}{\nu}$. Take the EM sequence $\Lambda^{(t+1)}=\arg\max_{\Lambda'\in B_r(\Lambda^*)\cap \Theta_n}Q_n(\Lambda'|\Lambda^{(t)})$ and $\Lambda^0\in B_r(\Lambda^*)\cap \Theta$.  Then if the sample size is large enough that $2^L n^{-1/3}\leq (1-\kappa)r$, then w.p. at least $1-u_L$
\begin{align*}
    \Vert \Lambda^{(t)}-\Lambda^*\Vert &\leq \kappa^t\Vert \Lambda^0-\Lambda^*\Vert+\frac{1}{1-\kappa}2^L n^{-1/3}.
\end{align*}
\end{theorem}
See section \ref{proof:sample-convergence} in the supplementary material for a proof. The proof follows that of Theorem 5 of \cite{balakrishnan2017statistical} closely. Note that this immediately implies that our algorithm can recover the true population MLE in large samples. Further by \cite{wellner2000two}, the population MLE is the true mean function under assumptions \ref{assm:counting-obs-process-ind},\ref{assm:bounded-mean-function} and \ref{assm:unif-bounded-mgf} even under Poisson process violations. Thus our method is robust to Poisson process violations.


\section{Experiments}

Here we show numerical and real data results to illustrate the utility of the proposed functional EM algorithm. In the M-step, we can choose any reasonable mean function estimator. In our experiments, we use a general likelihood-based augmented estimating equation (AEE) method \cite{wang2013augmented}, which uses monotone step functions when obtaining the mean function estimate with complete data. \cite{wang2011fitting,chiou2019semiparametric} provide software implementations of AEE and a wide variety of other panel count methods. First, we perform synthetic experiments that demonstrate accurate recovery of the true mean function (see Section \ref{sec:synthetic-appendix} of the supplemtary material). In particular, we also demonstrate good recovery of the true mean function using simulated data based on mixed poisson processes (which violate the Poisson process assumption) hence confirming theoretical robustness of functional EM to Poisson process assumption. In the following, we focus on two experiments involving real data. 




%



\subsection{Real Data with Synthetic Missingness}

We analyze blaTum (bladder tumor dataset) \cite{byar1980veterans} with 85 patients and counts of tumors taken at appointment times. We artificially delete intervals completely at random with probability $0.2$.  We then initialize $\Lambda^{(0)}$ by replacing the missing data with ${\sf Poisson}(1)$ random variables and fitting AEE.  We bootstrap 1,000 times, and plot the sample mean of our learned mean functions under complete data. We also set counts to zero in intervals with missing counts, which biases the mean function estimates. Figure \ref{fig:results}a compares inference from complete vs partially missing data using our wrapper vs zeroing out missing data.  Our wrapper learns a model much closer to the complete data than its initialization or the zeroing model. Section \ref{sec:bladder-tumor-appendix} of the supplementary material has additional experiments: we investigate sensitivity to initialization and different missingness probabilities. We also replace AEE with other M-step methods in Section \ref{sec:bladder-tumor-appendix} of the supplementary material and again confirm the utility of functional EM in approximating the mean function. The actual estimates vary slightly. In practice we recommend comparing results when making scientific conclusions; similar results from a wide variety of flexible methods indicate more robust scientific evidence.

\begin{figure}
    \centering
    \includegraphics[width=0.5\textwidth]{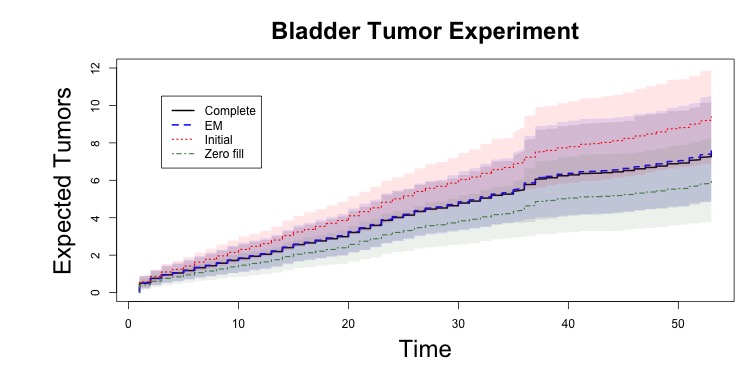}\includegraphics[width=0.5\textwidth]{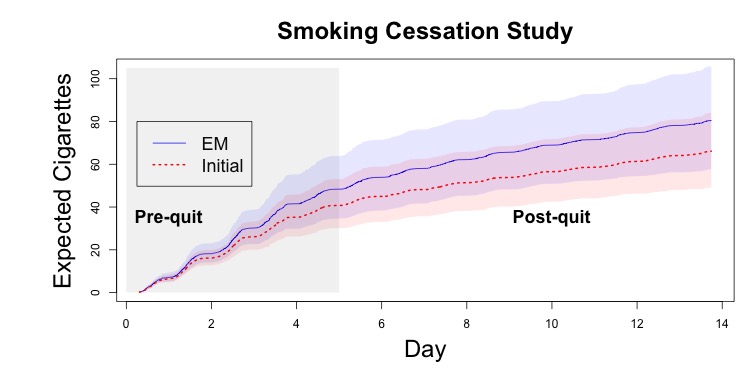}
    \caption{The estimated mean function and the 95\% bootstrap confidence band from \textbf{a)} bladder tumor dataset: synthetic 20\% missingness probability.  Initialization has missing data set to Poisson$(1)$ values.  \textbf{b)} smoking cessation study.  Initialization treats self-reported counts in intervals over 24 hours (about 5.6\% of observations) as valid.  }
    \label{fig:results}
\end{figure}


\subsection{Smoking Cessation Study}


We analyze data from an ongoing smoking cessation study in which the following EMA question was delivered randomly 3-4 times per day: "\textit{Since the last assessment, how many cigarettes did you smoke?}"  We have 125 participants tracked over 14 days, with 3-4 random EMAs targeted. The first five days are their normal smoking behavior, while the remaining days involve them attempting to quit. In both curves the smoking rates are much higher pre-quit than post-quit, suggesting that they do in fact attempt to quit. After discussion with psychologists, we treat counts in intervals of longer than 24 hours as missing because they are considered to be unreliable. We use the counts to initialize our model.  Unlike the previous experiments, we lack a ground truth. Note that this study is currently ongoing and we cannot link to it or share the dataset.

Figure \ref{fig:results}b shows the results based on $1,000$ non-parametric bootstrap samples and two models: one treating long intervals as valid (``initial"), and the other treating them as missing where we use EM.  Long intervals make up 5.6\% of observations.  The EM algorithm estimates that on average smokers attempting to quit smoked $80.4$ cigarettes by the end of the study; in contrast, AEE underestimated ($66.11$), a difference of 21.7\%. This is consistent with our collaborating behavioral scientists' hypothesis that participants may under-report a cigarette count when the gap between completed EMAs is large. The proposed functional EM is able to borrow count information across multiple shorter overlapping intervals from other participants to alleviate this under-reporting. We discuss other analyses of the dataset further in section \ref{sec:smoking-cessation-appendix} of the supplementary material.



\section{Discussion}

This paper proposed a functional EM algorithm to estimate the mean function for incomplete panel count data. Extending existing EM algorithm analysis to general non-parametric settings, we provided finite sample convergence guarantees to the truth. We conducted extensive experiments to illustrate the effectiveness of the proposed algorithm in recovering the true mean function. We applied the proposed algorithm to a smoking cessation study and found that participants may underestimate their cigarettes smoked over intervals longer than 24 hours. To the best of our knowledge, we are the first to apply panel count data methods to EMAs, despite their central role in behavioral science research. 

From a theoretical perspective, the main limitation is the MCAR assumption. Deviations from MCAR can happen if non-reporting depends on a subject's emotional state, which may be related to their smoking counts. Extensions to other missingness mechanisms such as Missing at Random (MAR) warrant future research; our analysis provides a general framework and first step, which can eventually be extended to also incorporate covariates under a more complex missing data model. Like some existing theoretical analysis of EM algorithms \cite{wu2016convergence}, another limitation is the need to optimize over a ball around the true mean function, which is unknown. Relaxing this condition is likely very challenging but also important future work. Finally, we would like test statistics for our estimator. This is challenging as the differences between fitted and true mean functions in norm do not converge to a normal distribution.  However, there are test statistics based on weighting \cite{balakrishnan2009new} that could be extended to this setting.

\section*{Broader Impact}

Understanding the dynamics for individuals who attempt to change and maintain behaviors to improve health has important societal value, for example, a comprehensive understanding of how smokers attempt to quit smoking may guide behavioral scientists to design better intervention strategies that tailor to the highest risk windows of relapse. Our theory and method provide a valid approach to understanding a particular aspect of the smoking behavior (mean function). The resulting algoithm is robust, readily adaptable and simple to implement, highlighting the potential for its wider adoption. The negative use case could be lack of sensitivity analysis around the assumptions such as missing data mechanism which may lead to misleading conclusions. Our current recommendation is to consult scientists about the plausibility of the assumption about missing data.

\bibliography{bibliography}
\bibliographystyle{acm}

\clearpage
\onecolumn
\appendix
\section{Proofs Related to Theorem \ref{thm:population-convergence}}\label{proof:population-convergence-general}

\subsection{Proof of Equality in Equation \ref{eqn:inner-product}}\label{proof:inner-product}
We start with the numerator

\begin{align*}
    Q(\Lambda^{(l)}+\eta\Lambda'|\Lambda)-Q(\Lambda^{(l)}|\Lambda)&=P(\sum_{j=1}^K \Delta N_j^{\tau_j}\Delta \Lambda_j^{s_j}\log(\Delta \Lambda^{(l)}_j+\eta\Delta \Lambda'_j)-(\Lambda^l_K+\eta\Lambda'_K))\\
    &\qquad-P(\sum_{j=1}^K \Delta N_j^{\tau_j}\Delta \Lambda_j^{s_j}\log(\Delta \Lambda^{(l)}_j)-(\Lambda^l_K)))\\
    &=P(\sum_{j=1}^K \Delta N_j^{\tau_j}\Delta \Lambda_j^{s_j}\log \frac{\Delta \Lambda^{(l)}_j+\eta\Delta \Lambda'_j}{\Delta \Lambda^{(l)}_j}-\eta \Lambda'_K)
\end{align*}

Now consider
\begin{align}
    \lim_{\eta\downarrow 0}\frac{\log \frac{\Delta \Lambda^{(l)}_j+\eta\Delta \Lambda'_j}{\Delta \Lambda^{(l)}_j}}{\eta}&=\lim_{\eta\downarrow 0}\log \left(\frac{\Delta \Lambda^{(l)}_j+\eta\Delta \Lambda'_j}{\Delta \Lambda^{(l)}_j}\right)^{1/\eta}\\
    &=\lim_{\eta\downarrow 0}\log\left(1+\eta\frac{\Delta \Lambda'_j}{\Delta \Lambda^{(l)}_j}\right)^{1/\eta}\\
    &=\log \exp(\frac{\Delta \Lambda'_j}{\Delta \Lambda^{(l)}_j})\\
    &=\frac{\Delta \Lambda'_j}{\Delta \Lambda^{(l)}_j}
\end{align}
We next need to show that we can pull the limit as $\eta\downarrow 0$ under integrals.  There are two relevant terms: $\frac{P(\sum_{j=1}^K \eta \Lambda'_K)}{\eta}$, which we can trivially handle by pulling $\eta$ outside the integral, and
\begin{align}
    \frac{P(\sum_{j=1}^K \Delta N_j^{\tau_j}\Delta \Lambda_j^{s_j}\log \frac{\Delta \Lambda^{(l)}_j+\eta\Delta \Lambda'_j}{\Delta \Lambda^{(l)}_j})}{\eta}&=P(\sum_{j=1}^K \Delta N_j^{\tau_j}\Delta \Lambda_j^{s_j}\frac{1}{\eta} \log \frac{\Delta \Lambda^{(l)}_j+\eta\Delta \Lambda'_j}{\Delta \Lambda^{(l)}_j})
\end{align}
Noting that $\frac{1}{\eta} \log \frac{\Delta \Lambda^{(l)}_j+\eta\Delta \Lambda'_j}{\Delta \Lambda^{(l)}_j}$ is monotone increasing to $\frac{\Delta \Lambda'_j}{\Delta \Lambda^{(l)}_j}$ as $\eta \downarrow 0$, we can apply the monotone convergence theorem to pull the limit under the integral.  Then
\begin{align}
    \lim_{\eta\downarrow 0}\frac{Q(\Lambda^{(l)}+\eta \Lambda'|\Lambda)-Q(\Lambda^{(l)}|\Lambda)}{\eta}&=P(\sum_{j=1}^K \Delta N_j^{\tau_j}\Delta \Lambda_j^{s_j}\frac{\Delta \Lambda'_j}{\Delta \Lambda^{(l)}_j}-\Lambda'_K)\\
    &=P(\sum_{j=1}^K (\frac{\Delta N_j^{\tau_j}\Delta\Lambda_j^{s_j}}{\Delta \Lambda^l_{j}}-1)(\Delta\Lambda'_{j}))\numberthis
\end{align}

\subsection{Proof of Lemma \ref{lemma:gradient-stability}}\label{proof:gradient-stability}
\begin{align*}
    \left \vert \sum_{j=1}^K s_j \left(\frac{\Delta \Lambda_{j}-\Delta \Lambda^*_{j}}{\Delta\Lambda^*_{j}}\right)\left(\Delta \Lambda'_{j}-\Delta\Lambda^*_{j}\right)\right\vert&\leq k_0\left\vert\frac{(\Delta \Lambda_j-\Delta\Lambda_j^*)(\Delta \Lambda_j'-\Delta \Lambda_j^*))}{\Delta \Lambda^*_j}\right\vert\\
    &\leq \frac{k_0}{c}\vert (\Delta \Lambda_j-\Delta\Lambda_j^*)(\Delta \Lambda_j'-\Delta \Lambda_j^*))\vert\\
    &\leq k_0 \frac{b^2}{c}\\
    &<\infty
\end{align*}
then any integral of this over $\mathcal{T}\times \mathcal{K}\times \mathcal{Z}$ will also be finite, and we can apply Fubini's theorem to such integrals. Then recalling that $\epsilon>0$ is the MCAR probability (assumption \ref{assm:MCAR}),
\begin{align*}
    \langle \nabla Q(\Lambda^*|\Lambda)-\nabla Q(\Lambda^*|\Lambda^*),\Lambda'-\Lambda^*\rangle&=P\left[\sum_{j=1}^K s_j \left( \frac{\Delta\Lambda _{j}-\Delta\Lambda^*_{j}}{\Delta \Lambda^*_{j}}\right)\left(\Delta \Lambda'_{j}-\Delta \Lambda^*_{j}\right)\right]\\
    &\leq \epsilon \left[P\left(\sum_{j=1}^K  \left( \frac{\Delta\Lambda _{j}-\Delta\Lambda^*_{j}}{\Delta \Lambda^*_{j}}\right)^2\right)\right]^{1/2}\\
    &\qquad\times\left[P\left(\sum_{j=1}^K  \left(\Delta \Lambda'_{j}-\Delta \Lambda^*_{j}\right)^2\right)\right]^{1/2}
\end{align*}
Here we used Fubini's theorem to pull out $\epsilon$.  Applying the CS inequality for inner products in $l^2$, and finally applying the CS inequality for expectations gives us the result.  The first term on rhs is then:
\begin{align*}
&\left[P\left(\sum_{j=1}^K \left( \frac{\Delta\Lambda _{j}-\Delta\Lambda^*_{j}}{\Delta \Lambda^*_{j}}\right)^2\right)\right]^{1/2}\leq \frac{1}{c} \left[P\left(\sum_{j=1}^K  \left(\Delta\Lambda _{j}-\Delta\Lambda^*_{j} \right)^2\right)\right]^{1/2}
\end{align*}
This gives us
\begin{align*}
    &P\left[\sum_{j=1}^K s_j \left( \frac{\Delta\Lambda _{j}-\Delta\Lambda^*_{j}}{\Delta \Lambda^*_{j}}\right)\left(\Delta \Lambda'_{j}-\Delta \Lambda^*_{j}\right)\right]\leq \frac{\epsilon}{c}\Vert \Lambda-\Lambda^*\Vert \Vert\Lambda'-\Lambda^*\Vert
\end{align*}
and thus
\begin{align*}
    \langle \nabla Q(\Lambda^*|\Lambda)-\nabla Q(\Lambda^*|\Lambda^*),\Lambda'-\Lambda^*\rangle &\leq  \frac{\epsilon}{c} \Vert \Lambda-\Lambda^*\Vert \Vert \Lambda'-\Lambda^*\Vert 
\end{align*}


\subsection{Proof of Lemma \ref{lemma:almost-strong-concavity}}\label{proof:almost-strong-concavity}
Note that
\begin{align}
    Q(\Lambda'|\Lambda)&=P\left(\sum_{j=1}^K \Delta N_j^{\tau_j}\Delta \Lambda^{s_j}_j\log[\Delta \Lambda'_j]-\Lambda'(T_K)\right)\\
    &=P\left(\sum_{j=1}^K\tau_j[\Delta N_j\log[\Delta \Lambda'_j]-\Lambda'(T_K)]\right)+P\left(\sum_{j=1}^K s_j[\Delta \Lambda_j\log[\Delta \Lambda'_j]-\Lambda'(T_K)]\right)
\end{align}

Define
\begin{align}
    Q_1(\Lambda'|\Lambda)
    &=P\left(\sum_{j=1}^K\Delta N_j\log[\Delta \Lambda'_j]-\Lambda'(T_K)\right)\\
    Q_2(\Lambda'|\Lambda)&= P\left(\sum_{j=1}^K\Delta \Lambda_j\log[\Delta \Lambda'_j]-\Lambda'(T_K)\right)
\end{align}
and note that if $\sum_{j=1}^K\tau_j[\Delta N_j\log[\Delta \Lambda'_j]-\Lambda'(T_j)]$ and $\sum_{j=1}^K s_j[\Delta N_j\log[\Delta \Lambda'_j]-\Lambda'(T_j)]$ are integrable over $\mathcal{N}\times \mathcal{T}\times \mathcal{K}\times\mathcal{Z}$ with respect to the measure $P$, we can apply Fubini's theorem to obtain
\begin{align}
    Q(\Lambda'|\Lambda)&=(1-\epsilon)Q_1(\Lambda'|\Lambda)+\epsilon Q_2(\Lambda'|\Lambda)
\end{align}

This proof takes place in four parts. We first show an inequality that allows us to characterize $r$ for $B_r(\Lambda^*)$. Next we show that in this ball the integrability conditions above hold. We then show that $Q_1$ and $Q_2$ are both strictly concave. Finally we show that $Q_1$ is strongly concave and thus $Q$, a positive linear combination of a strongly concave and a strictly concave function, is strongly concave.
\subsubsection{Characterizing the Radius of Contraction}
\begin{claim}\label{claim:term-I-inequality}
Let $h(x)=x(\log(x)-1)+1$.  For $\Vert \Lambda^*-\Lambda'\Vert_\infty\leq\frac{c}{4}$
\begin{align}
    h(\frac{\Delta \Lambda^*_j}{\Delta \Lambda'_j})\geq \frac{1}{3}(\frac{\Delta \Lambda^*_j}{\Delta \Lambda'_j}-1)^2
\end{align}
\end{claim}
\begin{proof}
Let $\phi(y)=(1+y)(\log(1+y)-1)+1$.  Then we can Taylor expand
\begin{align*}
    \phi(y)&=\phi(0)+\phi'(0)y+\frac{\phi''(\xi)}{2}y^2
\end{align*}
where $\xi$ lies between $0$ and $y$.  Then noting that $\phi(0)=\phi'(0)=0$ and $\phi''(y)=\frac{1}{1+y}$,
\begin{align*}
    \phi(y)&=\frac{\phi''(\xi)}{2}y^2\\
    &=\frac{1}{2(1+\xi)}y^2\\
    &\geq \frac{1}{3}y^2
\end{align*}
for $y\in (-1,1)$.  Letting $h(x)=x(\log x-1)+1$ this implies that $h(x)\geq \frac{1}{3}(x-1)^2$ for $(x-1)\in (-1,1)$ i.e. $x\in (0,2)$. Let $x=\frac{\Delta \Lambda^*_{j}}{\Delta\Lambda'_{j}}$.  Then we need $\Delta \Lambda_j'\geq \frac{1}{2}\Delta \Lambda_j^*$ for $h(x)\geq \frac{1}{3}(x-1)^2$ to hold. Note that if
\begin{align}\label{eqn:z-bound}
    |\Delta \Lambda^*_j-\Delta\Lambda'_j|\leq z\Delta \Lambda^*_j
\end{align}
Then
\begin{align*}
    \Delta \Lambda^*_j-\Delta\Lambda'_j&\leq z\Delta \Lambda^*_j\\
    (1-z)\Delta\Lambda^*_j&\leq \Delta \Lambda'_j\\
    \Delta \Lambda^*_j&\leq\frac{1}{1-z}\Delta \Lambda' _j
\end{align*}
and if $z=0.5$ the desired result holds. For equation \ref{eqn:z-bound} to hold it suffices to have $(\Delta \Lambda^*_j-\Delta\Lambda'_j)^2\leq z^2\Delta \Lambda^{*2}_j\leq z^2c^2$. Now noting that $(a-b)^2\leq 2(a^2+b^2)$,
\begin{align*}
    (\Delta\Lambda_j^*-\Delta\Lambda'_j)^2&=([\Lambda_j^*-\Lambda_j']-[\Lambda_{j-1}^*-\Lambda_{j'-1}])^2\\
    &\leq 2([\Lambda_j^*-\Lambda_j]^2+[\Lambda_{j-1}^*-\Lambda_{j-1}]^2)\\
    &\leq 4\Vert \Lambda^*-\Lambda'\Vert_\infty^2
\end{align*}
w.p. 1, and noting that we want $z=0.5$ a sufficient condition is
\begin{align*}
    4\Vert \Lambda^*-\Lambda'\Vert_\infty^2&\leq z^2c^2\\
    \Vert \Lambda^*-\Lambda'\Vert_\infty&\leq\frac{c}{4}
\end{align*}

\end{proof}
\subsection{Integrability Conditions}
First note by assumption \ref{assm:unif-bounded-mgf}, all moments of $N(\tau)$ are uniformly bounded and thus $\Vert N(\tau)\Vert_\infty<\infty$. Then for $\Lambda,\Lambda'\in B_r(\Lambda^*)$ and by assumption \ref{assm:bounded-number-observations},
\begin{align*}
    P\left\vert\sum_{j=1}^K\tau_j[\Delta N_j\log[\Delta \Lambda'_j]-\Lambda'(T_j)]\right\vert&\leq k_0\Vert N(\tau)\Vert_\infty \max(|\log \frac{c}{2}|,\log b)+k_0 b\\
    &<\infty\\
    P\left\vert\sum_{j=1}^K s_j[\Delta \Lambda_j\log[\Delta \Lambda'_j]-\Lambda'(T_j)]\right\vert&\leq k_0 b\max(|\log\frac{c}{2}|,\log b)+k_0b\\
    &<\infty
\end{align*}
\subsubsection{Strict Concavity of $Q_1$ and $Q_2$}\label{sec:strict-concavity}
For $Q_1$ we have
\begin{align}
    Q_1(\frac{\Lambda_1+\Lambda_2}{2}|\Lambda)&=P\left(\sum_{j=1}^K\Delta N_j\log[\frac{\Delta\Lambda_{1,j}+\Delta\Lambda_{2,j}}{2}]-\frac{\Lambda_{1,K}+\Lambda_{2,K}}{2}\right)\\
    &>P\left(\sum_{j=1}^K\Delta N_j\log[\sqrt{\Delta \Lambda_{1,j}}\sqrt{\Delta\Lambda_{2,j}}]-\frac{\Lambda_{1,K}+\Lambda_{2,K}}{2}\right)\textrm{ AM-GM inequality}\\
    &=\frac{1}{2}Q_1(\Lambda_1|\Lambda)+\frac{1}{2}Q_1(\Lambda_2|\Lambda)
\end{align}

and for $Q_2$ the same argument can be made.
\subsubsection{Strong Concavity of $Q$}
Now note that
\begin{align*}
    Q_1(\Lambda^*|\Lambda)&-Q_1(\Lambda'|\Lambda)+\langle\nabla Q_1(\Lambda^*|\Lambda),\Lambda'-\Lambda^*\rangle\\
    &=P\left(\sum_{j=1}^K [\Delta N_j\log\frac{\Delta\Lambda_j^*}{\Delta\Lambda_j'}-(\Delta\Lambda^*_j-\Delta\Lambda'_j)(1-(\frac{\Delta N_j}{\Delta \Lambda^*_j}-1))]\right)\\
    &=P\left(\sum_{j=1}^K [\Delta \Lambda_j^*\log[\frac{\Delta \Lambda^*_j}{\Delta \Lambda'_j}]-(\Delta\Lambda^*_j-\Delta\Lambda'_j)] \right)\\
    &=P\left(\sum_{j=1}^K \Delta\Lambda'_{j}\left(\frac{\Delta\Lambda^*_{j}}{\Delta\Lambda'_{j}}\log\frac{\Delta\Lambda^*_{j}}{\Delta\Lambda'_{j}}-(\frac{\Delta\Lambda^*_{j}}{\Delta\Lambda'_{j}}-1)\right)\right)\\
    &=P\left(\sum_{j=1}^K \Delta\Lambda'_j h(\frac{\Delta \Lambda^*}{\Delta \Lambda'_j})\right)\\
    &\geq \frac{1}{3} P\left(\sum_{j=1}^K\Delta \Lambda'_{j}(\frac{\Delta \Lambda^*_{j}}{\Delta \Lambda'_{j}}-1)^2\right)\\
    &\geq \frac{1}{3b}\Vert \Lambda'-\Lambda^*\Vert^2\\
    &\geq \frac{1}{3b}\Vert \Lambda'-\Lambda^*\Vert^2
\end{align*}
where the first inequality holds by claim \ref{claim:term-I-inequality}. Now by strict concavity of $Q_2$ as shown in \ref{sec:strict-concavity} we have
\begin{align}
    Q_2(\Lambda^*|\Lambda)-Q_2(\Lambda'|\Lambda)+\langle \nabla Q_2(\Lambda^*|\Lambda),\Lambda'-\Lambda^*\rangle\geq 0
\end{align}
Summing $(1-\epsilon)Q_1$ and $\epsilon Q_2$ we obtain
\begin{align}
    Q(\Lambda^*|\Lambda)-Q(\Lambda'|\Lambda)+\langle\nabla Q(\Lambda^*|\Lambda),\Lambda'-\Lambda^*\rangle&\geq \frac{(1-\epsilon)}{3b}\Vert \Lambda'-\Lambda^*\Vert^2
\end{align}

\subsection{Proof of Population Contractivity}\label{proof:population-contractivity}
We now state the main proof of population contractivity.  Denote
\begin{align*}
    V(\Lambda'|\Lambda)&=Q(\Lambda'|\Lambda)-Q(\Lambda^*|\Lambda)-\langle \nabla Q(\Lambda^*|\Lambda),\Lambda'-\Lambda^*\rangle 
\end{align*}
then
\begin{align*}
    0&\leq Q(\Lambda'|\Lambda)-Q(\Lambda^*|\Lambda)\\
    &=V(\Lambda'|\Lambda)+\langle \nabla Q(\Lambda^*|\Lambda),\Lambda'-\Lambda^*\rangle\\
    &= V(\Lambda'|\Lambda)+\langle \nabla Q(\Lambda^*|\Lambda)-\nabla Q(\Lambda^*|\Lambda^*),\Lambda'-\Lambda^*\rangle+\langle \nabla Q(\Lambda^*|\Lambda^*),\Lambda'-\Lambda^*\rangle\\
    &\leq V(\Lambda'|\Lambda)+\langle \nabla Q(\Lambda^*|\Lambda)-\nabla Q(\Lambda^*|\Lambda^*),\Lambda'-\Lambda^*\rangle\textrm{ KKT conditions}\\
    &\leq -\nu \Vert \Lambda'-\Lambda^*\Vert^2+\gamma\Vert \Lambda-\Lambda^*\Vert \Vert \Lambda'-\Lambda^*\Vert\textrm{ technical Lemmas}
\end{align*}
and rearranging terms and dividing both sides by $\Vert \Lambda'-\Lambda^*\Vert$ gives the desired result.  Note that we used $\langle \nabla Q(\Lambda^*|\Lambda^*),\Lambda'-\Lambda^*\rangle\leq 0$, which if $\langle \cdot,\cdot\rangle$ were a valid inner product would be the KKT conditions.  However since $\langle\cdot,\cdot\rangle$  may not be a valid inner product, they must be checked specifically.  \cite{wellner2000two} does it in the sample case for the true log-likelihood: it is easy to verify that it still holds in the population case for $Q$-functions.  Noting that $\Lambda^*$ maximizes $Q(\Lambda'|\Lambda^*)$, we have that $Q(\Lambda^{*}+\eta( \Lambda'-\Lambda^*)|\Lambda^*)-Q(\Lambda^*|\Lambda^*)\leq 0$
\begin{align*}
    \langle \nabla Q(\Lambda^*|\Lambda^*),\Lambda'-\Lambda^*\rangle&=\lim_{\eta\downarrow 0}\frac{Q(\Lambda^{*}+\eta( \Lambda'-\Lambda^*)|\Lambda^*)-Q(\Lambda^*|\Lambda*)}{\eta}\\
    &\leq 0
\end{align*}

\subsection{Proof of Theorem \ref{thm:population-convergence}}\label{proof:population-convergence}

By induction.  It holds for $t=0$.  Assume it holds for $t\geq 0$.  Then $\Lambda^{(t+1)} \in B_r(\Lambda^*)$ and by assumption $Q(\Lambda^{(t+1)}|\Lambda^{(t)})\geq Q(\Lambda^*|\Lambda^{(t)})$.  Applying population contractivity and the induction assumption,
\begin{align*}
    \Vert \Lambda^{t+1}-\Lambda^*\Vert &\leq \frac{\gamma}{\nu}\Vert \Lambda^t-\Lambda^*\Vert\\
    &\leq \left(\frac{\gamma}{\nu}\right)^{t+1} \Vert \Lambda^0-\Lambda^*\Vert 
\end{align*}

\section{Proofs Related to Theorem \ref{thm:sample-convergence}}

\subsection{Proof of Proposition \ref{prop:rate-of-convergence}}\label{proof:rate-of-convergence}
\subsubsection{Definitions and Background from the Literature}

Before proving the proposition, we restate several important results that we use from existing literature. We repeat two definitions and two theorems from \cite{sen2018gentle}, adjusted to our notation. Note that $\eta$-brackets are normally called $\epsilon$-brackets. However, since we have already used $\epsilon$ to denote MCAR probabilities, we call them $\eta$-brackets.
\begin{definition}($\eta$-bracket) Let $(\mathcal{F},d)$ be a normed space of functions distance metric $d$ induced by some norm. Given two functions $l(\cdot)$ and $g(\cdot)$, the bracket $[l,g]$ is the set of all functions $f \in \mathcal{F}$ with $l(u)\leq f(u)\leq g(u)\forall u\in [0,\tau]$. An $\eta$-bracket is a bracket $[l,g]$ with $d(l,g)<\eta$.
\end{definition}
\begin{definition}
(Bracketing numbers). The bracketing number $N_{[]}(\eta,\mathcal{F},L_2(P))$ is the minimum number of $\eta$-brackets needed to cover $\mathcal{F}$ using $L^2(P)$ distance.
\end{definition}
\begin{definition}
(Bracketing Integral) The bracketing integral is defined as
\begin{align}
    J_{[]}(\delta,\mathcal{F},L_2(P))&\equiv \int_0^\delta \sqrt{\log(N_{[]}(\eta,\mathcal{F}\cup\{0\},L_2(P)))}d\eta
\end{align}
\end{definition}
Note that since any non-empty set requires at least one bracket to cover it and $\log(x+1)\leq 1+\log(x)$ for $x\geq 1$, 
\begin{align*}
    \int_0^\delta \sqrt{\log(N_{[]}(\eta,\mathcal{F}\cup\{0\},L_2(P)))}\leq \int_0^\delta \sqrt{1+\log(N_{[]}(\eta,\mathcal{F},L_2(P)))}
\end{align*}
the right hand side is sometimes used as the definition of the bracketing integral, but we use the left hand side, following \cite{sen2018gentle}. We now restate a theorem from \cite{sen2018gentle}, with the notation heavily adapted to our setting for clarity. The theorem is otherwise the same.
\begin{theorem}
\label{thm:5.1}(Theorem 5.1 in \cite{sen2018gentle}) Let $(\Theta\cap B_r(\Lambda^*),\Vert \cdot\Vert)$ be a semi-metric space. Fix $n\geq 1$. Let $\{Q_n(\Lambda'|\Lambda):\Lambda'\in \Theta \cap B_r(\Lambda^*)\}$ be a stochastic process and $\{Q(\Lambda'|\Lambda):\Lambda'\in \Theta \cap B_r(\Lambda^*)\}$ be a deterministic process. Assume
\begin{align}\label{eqn:sep-condition}
    Q(\Lambda'|\Lambda)-Q(M(\Lambda)|\Lambda)&\leq -c_1 \Vert \Lambda'-M(\Lambda)\Vert^2 
\end{align}
for some $c_1>0$. We call this the separation condition. Further, let
\begin{align}
    U_n(\Lambda'|\Lambda)&=Q_n(\Lambda'|\Lambda)-Q(\Lambda'|\Lambda)
\end{align}
and assume that there exists some function $\phi_n(\cdot)$ satisfying the following three conditions
\begin{enumerate}
    \item The following expected supremum condition holds
    \begin{align}\label{eqn:expected-supremum}
    \mathbb{E}\left[\sup_{\Lambda':\Vert \Lambda'-M(\Lambda)\Vert \leq \delta }\sqrt{n}\vert U_n(\Lambda'|\Lambda)-U_n(M(\Lambda)|\Lambda)\vert\right] \lesssim \phi_n(\delta)
    \end{align}
    \item there exists $\alpha<2$ so that
\begin{align}\label{eqn:function-pull-out-constant}
    \phi_n(dx)&\leq d^\alpha\phi_n(x)\forall d>1,x>0
\end{align}
    \item for the rate of convergence $r_n$
    \begin{align}\label{eqn:rate-inequality}
    \phi_n(r_n)&\lesssim \sqrt{n}r_n^2
\end{align}
as $n$ varies. 
\end{enumerate}
Here $\lesssim$ means $\leq$ the right hand side times a constant. Then for every $L>0$, $\Vert M_n(\Lambda)-M(\Lambda)\Vert\leq 2^L r_n$ with probability at least $1-u_L$. Here $u_L=\tilde{c}\sum_{j>M}2^{j(\alpha-2)}$, where $\tilde{c}$ only depends on the constants in the separation condition and the expected supremum bound.
\end{theorem}
Note that this is essentially a special case of Theorem 3.2.5 of \cite{van1996weak}, but it uses expectations instead of outer expectations. It makes the stronger version of their assumptions and draws a stronger conclusion, giving a finite sample bound. The key in our setting will be to ensure that  $\tilde{c}$ does not vary across iterations, which requires that the constants for the separation condition and the expected supremum bound do not vary across iterations. Importantly, we cannot always apply this theorem in the general functional EM setting: it requires that the sample $Q$-function at its maximizer over a Sieve is equal to the sample $Q$-function at its maximizer over the full function space. This holds for the Poisson process log-likelihood for panel count data for a range of cases: for instance with step functions or splines where jumps/knots occur at observation times. It does not necessarily hold for arbitrary models. In the latter case we may be able to use Theorem 6.1 of \cite{sen2018gentle}, but this would require checking it carefully for each potential model.

We also note
\begin{theorem}\label{thm:4.12}
(Theorem 4.12 of \cite{sen2018gentle}) For any class $\mathcal{F}$ of measurable functions $f:\mathcal{X}\rightarrow\mathbb{R}$ such that $Pf^2<\delta^2$ and $\Vert f\Vert_{\infty}\leq M$ for every $f$,
\begin{align}
    \mathbb{E}[\sup_{f\in \mathcal{F}}\vert\sqrt{n}(P_n-P)\vert]&\leq \tilde{K}J_{[]}(\delta,\mathcal{F},L_2(P)) \left(1+\frac{J_{[]}(\delta,\mathcal{F},L_2(P))}{\delta^2 \sqrt{n}}M\right)
\end{align}
where $\tilde{K}>0$ is some constant.
\end{theorem}
Importantly, $\tilde{K}$ is a \textit{universal} constant and does not depend on $\mathcal{F}$. This was noted by \cite{ma2005robust}. A version of this theorem was originally Lemma 3.4.2 in \cite{van1996weak}.


\subsubsection{Outline}

We follow \cite{balakrishnan2011class}, which proves the rate of convergence for the maximum pseudo-likelihood of a Poisson process objective function for panel count data: extending their proof to the expected complete data log-likelihood case is straightforward. However, we face the issue that we want a high probability uniform bound on the distance between sample and population M-steps \textit{across} EM iterations, whereas they neeeded an asymptotic high probability rate of convergence for the pseudo MLE. This poses three challenges: 1) our objective function at each iteration is the expected log-likelihood rather than the log-likelihood. Thus we cannot prove that the separation condition holds using the same techniques. 2) the separation condition must hold always rather than only in a neighborhood of the optimum 3) we need to check that constants are the same across EM iterations. Our aim is to apply Theorem \ref{thm:5.1} and show that the constant $\tilde{c}$ in $u_L$ does not vary across EM iterations. Note that other than checking the separation condition, the majority of this proof simply repeats the proof strategy of \cite{balakrishnan2011class} but fills in details of results they call to make sure that constants don't vary across iterations of EM.

This proof takes place in four parts. We first show that the separation condition holds. We then bound the expectation of the supremum of the magnitude of an empirical process. We next prove the two properties of the function involved in that bound to show the rate of convergence, and finally conclude by applying Theorem $\ref{thm:5.1}$.
\subsubsection{Separation Condition}
We first prove that the separation condition given by Equation \ref{eqn:sep-condition} holds. This involves applying functional second order Taylor expansions to $Q_1$ and $Q_2$ and using the remainder terms to obtain quadratic lower bounds.
\begin{claim}
For any $\Lambda,\Lambda'\in B_r(\Lambda^*)$,
\begin{align*}
    Q(M(\Lambda)|\Lambda)-Q(\Lambda'|\Lambda)&\geq \left((1-\epsilon)\frac{c}{b^2}+\epsilon\frac{c}{2b^2}\right)\Vert M(\Lambda)-\Lambda'\Vert^2
\end{align*}
\end{claim}
\begin{proof}
Consider the functional second order expansion of $Q(\Lambda'|\Lambda)$ at $Q(M(\Lambda)|\Lambda)$, which we can do by \cite{3076726}.
\begin{align*}
    Q_1(M(\Lambda)|\Lambda)-Q_1(\Lambda'|\Lambda)&=-\langle \nabla Q_1(M(\Lambda)|\Lambda),\Lambda'-M(\Lambda)\rangle\\
    &\qquad-P(-\sum_{j=1}^K(\Delta M(\Lambda)_j-\Delta \Lambda'_j)\frac{\Delta N_j}{\Delta \xi_{1,j}^2}(\Delta M(\Lambda)_j-\Delta\Lambda'_j))\\
    &= -\langle\nabla  Q_1(M(\Lambda)|\Lambda),\Lambda'-M(\Lambda)\rangle\\
    &\qquad+P(\sum_{j=1}^K(\Delta M(\Lambda)_j-\Delta \Lambda'_j)\frac{\Delta \Lambda_j^*}{\Delta \xi_{1,j}^2}(\Delta M(\Lambda)_j-\Delta\Lambda'_j))\\
    &\geq-\langle \nabla Q_1(M(\Lambda)|\Lambda),\Lambda'-M(\Lambda)\rangle+ \frac{c}{b^2}\Vert M(\Lambda)-\Lambda'\Vert^2
\end{align*}
here $\xi_{1,j}=\Lambda'+\eta_1(M(\Lambda)-\Lambda')$ for $\eta_1\in [0,1]$. Further,
\begin{align*}
    Q_2(M(\Lambda)|\Lambda)-Q_2(\Lambda'|\Lambda)&=-\langle\nabla Q_2(M(\Lambda)|\Lambda),\Lambda'-M(\Lambda)\rangle\\
    &\qquad-P(-\sum_{j=1}^K(\Delta M(\Lambda)_j-\Delta \Lambda'_j)\frac{\Delta \Lambda_j}{\Delta \xi_{2,j}^2}(\Delta M(\Lambda)_j-\Delta\Lambda'_j))\\
    &= -\langle\nabla Q_2(M(\Lambda)|\Lambda),\Lambda'-M(\Lambda)\rangle\\
    &\qquad+P(\sum_{j=1}^K(\Delta M(\Lambda)_j-\Delta \Lambda'_j)\frac{\Delta \Lambda_j}{\Delta \xi_{2,j}^2}(\Delta M(\Lambda)-\Delta \Lambda'))\\
    &\geq -\langle \nabla Q_2(M(\Lambda)|\Lambda),\Lambda'-M(\Lambda)\rangle+\frac{c}{2b^2}\Vert M(\Lambda)-\Lambda'\Vert^2
\end{align*}
where the last line follows since $\Lambda\in B_r(\Lambda^*)$ so that $\Vert \Lambda-\Lambda^*\Vert_\infty\leq \frac{c}{4}$ and thus $\Delta \Lambda\geq \frac{1}{2}\Delta \Lambda^*\geq \frac{1}{2}c$ w.p. 1. Noting that $\langle \nabla Q(M(\Lambda)|\Lambda),\Lambda'-M(\Lambda)\rangle \leq 0$ by the KKT conditions,
\begin{align*}
    Q(M(\Lambda)|\Lambda)-Q(\Lambda'|\Lambda)&\geq \left((1-\epsilon)\frac{c}{b^2}+\epsilon\frac{c}{2b^2}\right)\Vert M(\Lambda)-\Lambda'\Vert^2
\end{align*}
\end{proof}

and thus the separation condition holds since we optimized over $\Theta_n \cap B_r(\Lambda^*)$.
\subsubsection{Bounding the Expectation of the Supremum of the Magnitude of the Empirical Process}
Our aim in this section is to apply Theorem \ref{thm:4.12} and use the result to show that the expected supremum condition, Equation \ref{eqn:expected-supremum} in Theorem \ref{thm:5.1} holds. Let
\begin{align*}
    \Theta_\delta\equiv \{\Lambda':\Vert\Lambda'-M(\Lambda)\Vert\leq \delta,\Lambda'\in \Theta\cap B_r(\Lambda^*)\}
\end{align*}
Let $m_{\Lambda',\Lambda}(Y)\equiv \sum_{j=1}^K \Delta N_{j}^{\tau_j}\Delta \Lambda^{s_j}(T_j)\log[ \Delta\Lambda'_j]-\Lambda'(T_{K})$ and
\begin{align}
    \mathcal{M}_\delta&\equiv \{m_{\Lambda',\Lambda}(Y)-m_{M(\Lambda),\Lambda}(Y):\Lambda'\in \Theta_\delta\}
\end{align}
This section proceeds as follows. We first show that for all $f\in \mathcal{M}_{\delta}$, $Pf^2\leq c_2\delta^2$ for some constant $c_2>0$ and $\Vert f\Vert_{\infty}\leq c_3$ for some $c_3>0$. We next show a bound on the bracketing entropy in terms of the bracket size. We then use this to bound the bracketing integral using $\delta^{1/2}$. We combine all of this to bound the expectation of the supremum of interest. We must carefully note that relevant constants do not vary across iterations.

\begin{claim}
For $\Lambda'\in \Theta_\delta$, $P|m_{\Lambda',\Lambda}(Y)-m_{M(\Lambda),\Lambda}(Y)|^2\leq c_2\delta^2$ for some $c_2>0$ that does not depend on $\Lambda'$ or $\Lambda$.
\end{claim}
\begin{proof}
\begin{align*}
    P(m_{\Lambda',\Lambda}(Y)&-m_{M(\Lambda),\Lambda}(Y))^2\\
    &=P\left(\sum_{j=1}^K [\Delta N_j^{\tau_j}\Delta \Lambda_j^{s_j}\log \frac{\Delta \Lambda'_j}{\Delta M(\Lambda)_j}-(\Delta \Lambda'_j-\Delta M(\Lambda)_j)]\right)^2\\
    &\leq k_0P\left(\sum_{j=1}^K [\Delta N_j^{\tau_j}\Delta \Lambda_j^{s_j}\log \frac{\Delta \Lambda'_j}{\Delta M(\Lambda)_j}-(\Delta \Lambda'_j-\Delta M(\Lambda)_j)]^2\right)\\
    &\qquad\qquad\textrm{Cauchy Schwarz}\\
    &\leq 2k_0P\left(\sum_{j=1}^K [(\Delta N_j^{\tau_j}\Delta \Lambda_j^{s_j})^2(\log \frac{\Delta \Lambda'_j}{\Delta M(\Lambda)_j})^2+(\Delta \Lambda'_j-\Delta M(\Lambda)_j)^2]\right)\\
    &\qquad\qquad\textrm{since $(a-b)^2\leq 2(a^2+b^2)$}\\
    &=2k_0\left[P\left(\sum_{j=1}^K (\Delta N_j^{\tau_j}\Delta \Lambda_j^{s_j})^2(\log \frac{\Delta \Lambda'_j}{\Delta M(\Lambda)_j})^2\right)+\Vert \Lambda'-M(\Lambda)\Vert^2\right]\numberthis{}\label{eqn:bounding-squared}
\end{align*}
Then note
\begin{align*}
    P\left(\sum_{j=1}^K (\Delta N_j^{\tau_j}\Delta \Lambda_j^{s_j}\log \frac{\Delta \Lambda'_j}{\Delta M(\Lambda)_j})^2\right)&\leq P\left(\sum_{j=1}^K (\Delta N_j^{\tau_j}\Delta \Lambda_j^{s_j})^2\frac{(\Delta \Lambda'_j-\Delta M(\Lambda)_j)^2}{\min(\Delta \Lambda'_j,\Delta M(\Lambda)_j)^2}\right)\\
    &\leq \frac{1}{4c^2} P\left(\sum_{j=1}^K(\Delta N_j^{\tau_j}\Delta \Lambda_j^{s_j})^2(\Delta \Lambda'_j-\Delta M(\Lambda)_j)^2\right)\\
    &\leq \frac{\max(\Vert N(\tau)\Vert_\infty^2,b^2)}{4c^2}P\left(\sum_{j=1}^K (\Delta \Lambda'_j-\Delta M(\Lambda)_j)^2\right)\\
    &\qquad\qquad\textrm{ assumption \ref{assm:unif-bounded-mgf}}\\
    &=\frac{\max(\Vert N(\tau)\Vert_\infty^2,b^2)}{4c^2}\Vert \Lambda'-M(\Lambda)\Vert^2
\end{align*}
where the first line uses the inequality $1-\frac{1}{x} \leq \log(x) \leq x - 1$ which implies $\left(\log \left( \frac{x}{y} \right)\right)^2 \leq (x-y)^2 / \min(x,y)^2$.
Plugging this back into equation \ref{eqn:bounding-squared} we obtain
\begin{align*}
    P(m_{\Lambda',\Lambda}(Y)-m_{M(\Lambda),\Lambda}(Y))^2&\leq 2k_0\left[\frac{\max(\Vert N(\tau)\Vert_\infty^2,b^2)}{4c^2}\Vert \Lambda'-M(\Lambda)\Vert^2+\Vert \Lambda'-M(\Lambda)\Vert^2\right]\\
    &=\left(2k_0\frac{\max(\Vert N(\tau)\Vert_\infty^2,b^2)}{4c^2}+1\right)\Vert \Lambda'-M(\Lambda)\Vert^2\\
    &\leq \left(2k_0\frac{\max(\Vert N(\tau)\Vert_\infty^2,b^2)}{4c^2}+1\right)\delta^2
\end{align*}
so that we have $P|m_{\Lambda',\Lambda}(Y)-m_{M(\Lambda),\Lambda}(Y)|^2\leq c_2 \delta^2$ for some constant $c_2$, and $c_2$ does not depend on either $\Lambda'$ or $\Lambda$.
\end{proof}
\begin{claim}
For $\Lambda'\in \Theta_\delta$, $\Vert m_{\Lambda',\Lambda}(Y)-m_{M(\Lambda),\Lambda}(Y)\Vert_\infty\leq c_3$, where $c_3$ does not depend on $\Lambda$ or $\Lambda'$.
\end{claim}
\begin{proof}
Again using $\left\vert \log \left( \frac{x}{y} \right)\right\vert \leq \vert x-y\vert  / \vert \min(x,y)\vert$, we have
\begin{align*}
    \left\vert m_{\Lambda',\Lambda}(Y)-m_{M(\Lambda),\Lambda}(Y)\right\vert &=\left\vert \sum_{j=1}^K [\Delta N_j^{\tau_j}\Delta \Lambda_j^{s_j}\log \frac{\Delta \Lambda'_j}{\Delta M(\Lambda)_j}-(\Delta \Lambda'_j-\Delta M(\Lambda)_j)]\right\vert \\
    &\leq \left\vert \sum_{j=1}^K \max(\Vert N(\tau)\Vert_\infty,b)\frac{\Delta \Lambda'_j-\Delta M(\Lambda)_j}{\min(\Delta \Lambda'_j,\Delta M(\Lambda)_j)}\right\vert\\
    &\qquad+\left\vert \sum_{j=1}^K (\Delta \Lambda_j'-\Delta M(\Lambda)_j)\right\vert\\
    &\leq k_0\frac{\max(\Vert N(\tau)\Vert_\infty,b)}{2c}[2\Vert \Lambda'-M(\Lambda)\Vert_\infty ]+2k_0\Vert \Lambda'-\Lambda^*\Vert_\infty\\
    &\leq k_0\frac{\max(\Vert N(\tau)\Vert_\infty,b)}{2c}2[\Vert \Lambda'-\Lambda^*\Vert_\infty\\
    &\qquad+\Vert M(\Lambda)-\Lambda^*\Vert_\infty]+2k_0\Vert \Lambda'-\Lambda^*\Vert_\infty\\
    &\leq k_0 \frac{\max(\Vert N(\tau)\Vert_\infty,b)}{2}+k_0 \frac{c}{2}\\
    &=c_3
\end{align*}
where we used that $\Lambda',\Lambda^*,M(\Lambda)\in B_r(\Lambda^*)$ an $L^\infty$ ball with $r=\frac{c}{4}$
\end{proof}
By theorem 2.7.5 of \cite{van1996weak}, which bounds the bracketing number of monotone functions mapping to $[0,1]$, and noting that $\mathcal{M}_\delta$ has bracketing number less than or equal to that of $\Theta\cap B_r(\Lambda^*)$, which was shown in \cite{balakrishnan2009new},
\begin{align}
    \log N_{[]}(\eta,\mathcal{M}_\delta,L_2(P))\leq c_4\eta^{-1}
\end{align}
where $c_4$ only depends on $U_{\textrm{all}}$, the uniform upper bound in $\Theta$. Noting that $\mathcal{M}_{\delta}\cup \{0\}=\mathcal{M}_{\delta}$, we have
\begin{align}
    \int_0^\delta \sqrt{\log N_{[]}(\eta,\mathcal{M}_\delta\cup\{0\},L_2(P))}d\eta \leq c_5\delta^{1/2}
\end{align}
where again $c_5$ only depends on $U_\textrm{all}$. Let
\begin{align}
    \Vert \sqrt{n}(P_n-P)\Vert_{M_{\delta}}&=\sup_{f\in M_{\delta}}\vert \sqrt{n}(P_n-P)f\vert
\end{align}
and note that
\begin{align}
    \mathbb{E}\Vert \sqrt{n}(P_n-P)\Vert_{M_{\delta}}=\mathbb{E}\left[\sup_{\Lambda':\Vert \Lambda'-M(\Lambda)\Vert \leq \delta }\sqrt{n}\vert U_n(\Lambda'|\Lambda)-U_n(M(\Lambda)|\Lambda)\vert\right]
\end{align}
\begin{claim}
Because $P|m_{\Lambda',\Lambda}(Y)-m_{M(\Lambda),\Lambda}(Y)|^2\leq c_2\delta^2$ and $\Vert m_{\Lambda',\Lambda}(Y)-m_{M(\Lambda),\Lambda}(Y)\Vert_\infty\leq c_3$, we have
\begin{align}
\mathbb{E}\Vert \sqrt{n}(P_n-P)\Vert_{\mathcal{M}_\delta}\leq c_6 \phi_n(\delta)
\end{align}
for $\phi_n(\delta)=\delta^{1/2}+\delta^{-1}n^{-1/2}$, where $c_6$ does not depend on $\Lambda'$ or $\Lambda$.
\end{claim}
\begin{proof}
Here we use Theorem \ref{thm:4.12} in order to prove Equation \ref{eqn:expected-supremum} in Theorem \ref{thm:5.1}. By Theorem \ref{thm:4.12} and again noting that $\mathcal{M}_{\delta}\cup \{0\}=\mathcal{M}_{\delta}$, we have
\begin{align*}
    \mathbb{E}&\Vert \sqrt{n}(P_n-P)\Vert_{\mathcal{M}_\delta}\\
    &\leq \tilde{K}\left(\frac{(\int_0^{\sqrt{c_2}\delta} \sqrt{\log N_{[]}(\eta,\mathcal{M}_\delta,L_2(P))}d\eta)^2c_3}{c_2\delta^2 \sqrt{n}}+\int_0^{\sqrt{c_2}\delta} \sqrt{\log N_{[]}(\eta,\mathcal{M}_\delta,L_2(P))}d\eta\right)\\
    &\leq \tilde{K}\left(\frac{c_5^2(c_2^{1/2})c_3\delta}{c_2\delta^2\sqrt{n}}+c_5c_2^{1/4}\delta^{1/2}\right)\\
    &=\tilde{K}\left(\frac{c_5^2c_3}{c_2^{1/2}\delta\sqrt{n}}+c_5c_2^{1/4}\delta^{1/2}\right)\\
    &\leq \tilde{K}\max\left(\frac{c_5^2 c_3}{c_2^{1/2}},c_5c_2^{1/4}\right)\left(\frac{1}{\delta\sqrt{n}}+\delta^{1/2}\right)
\end{align*}
Set $c_6=\tilde{K}\max\left(\frac{c_5^2 c_3}{c_2^{1/2}},c_5c_2^{1/4}\right)$ and we are done.
\end{proof}
With this we have the bound on the expected supremum and have proven Equation \ref{eqn:expected-supremum} for Theorem \ref{thm:5.1}.

\subsubsection{Characterizing the Function in the Bound}
We first prove Equation \ref{eqn:function-pull-out-constant} in Theorem \ref{thm:5.1}.
\begin{align}
    \phi_n(d\delta)&=(d\delta)^{1/2}+\frac{1}{(d\delta)\sqrt{n}}\\
    &=\sqrt{d}\left(\delta^{1/2}+\frac{1}{d^{3/2}\delta\sqrt{n}}\right)\\
    &\leq \sqrt{d}\left(\delta^{1/2}+\frac{1}{\delta\sqrt{n}}\right)\textrm{ for }d\geq 1\\
    &=\sqrt{d}\phi_n(\delta)
\end{align}
and thus $\alpha=\frac{1}{2}$.

Next we prove Equation \ref{eqn:rate-inequality}. Let $r_n=n^{-1/3}$. Then
\begin{align}
    \phi_n(r_n)&=2n^{-1/6}
\end{align}
and
\begin{align}
    \sqrt{n}r_n^2&=\sqrt{n}(n^{-1/3})^2\\
    &=n^{-1/6}
\end{align}
and thus
\begin{align}
    \phi_n(r_n)\leq 2\sqrt{n}r_n^2
\end{align}

\subsubsection{Putting it All Together}
We have now proven that all of the assumptions of Theorem \ref{thm:5.1} hold, and that constants do not vary across iterations. Then for some constant $\tilde{c}$ which \textit{does not vary across iterations}, for any $L>0$ and using $\alpha=\frac{1}{2}$, w.p. $1-\tilde{c}\sum_{j>M}2^{-3j/2}$,
\begin{align}
    \Vert M_n(\Lambda)-M(\Lambda)\Vert\leq 2^L n^{-1/3}
\end{align}

\subsection{Proof of Theorem \ref{thm:sample-convergence}}\label{proof:sample-convergence}

Note that if $\forall \Lambda\in B_r(\Lambda^*)$, $\Vert M_n(\Lambda)-M(\Lambda)\Vert \leq 2^L n^{-1/3}$, then we have $\sup_{\Lambda\in B_r(\Lambda^*)}\Vert M_n(\Lambda)-M(\Lambda)\Vert \leq 2^L n^{-1/3}$. We claim for any $t>0$,
\begin{align*}
    \Vert \Lambda^{(t+1)}-\Lambda^*\Vert &\leq \kappa\Vert \Lambda^{(t)}-\Lambda^*\Vert+2^L n^{-1/3}
\end{align*}
We prove by induction.  First, this holds for $t=1$.
\begin{align*}
    \Vert \Lambda^{(1)}-\Lambda^*\Vert &\leq \Vert M(\Lambda^{(0)})-\Lambda^*\Vert+\Vert M_{n}(\Lambda^{(0)})-M(\Lambda^{(0)})\Vert\\
    &\leq \kappa\Vert \Lambda^{(0)}-\Lambda^*\Vert +2^Ln^{-1/3}
\end{align*}
Now assume holds true for $t>0$.  Then for $t+1$,
\begin{align*}
    \Vert \Lambda^{(t+1)}-\Lambda^*\Vert&\leq \Vert M(\Lambda^{(t)})-\Lambda^*\Vert+\Vert M_{n}(\Lambda^{(t)})-M(\Lambda^{(t)})\Vert\\
    &\leq \kappa\Vert \Lambda^{(t)}-\Lambda^*\Vert+2^Ln^{-1/3}
\end{align*}
Now iterating we have
\begin{align*}
    \Vert \Lambda^{(t)}-\Lambda^*\Vert &\leq \kappa^t\Vert \Lambda^{(0)}-\Lambda^*\Vert+\frac{1}{1-\kappa}2^Ln^{-1/3}
\end{align*}

\section{Synthetic Analysis}\label{sec:synthetic-appendix}

\subsection{Square Root Synthetic Experiment}

We generate synthetic panel count data from mixed inhomogeneous Poisson processes with conditional mean functions $\Lambda^*(u)|X=Xu^{1/2}$, $\Lambda^*(u)|X=Xu^2$, where $X\sim \textrm{uniform}(0,2)$. The mean functions are then $\Lambda^*(u)=u^{1/2}$ and $\Lambda^*(u)=u^2$, respectively. The counting process conditional on $X$ is Poisson, but the marginal counting process is not. We use $100$ trajectories, each with $30$ observations and for each observation set it to missing with probability $0.2$.  We initialize the mean function $\Lambda^{(0)}$ by replacing the missing data with ${\sf Poisson}(1)$ random variables and fitting a model.
We generate $1000$ Monte Carlo runs and create Monte Carlo marginal confidence intervals from those runs.

Fig \ref{fig:square-root} compares the true mean function against AEE wrapped with our method vs AEE directly on the corrupted data. Taking the corrupted data as given learns highly biased results, while wrapping AEE with our algorithm learns close to the true mean function for both experiments.

\begin{figure}
    \centering
    \includegraphics[width=0.5\textwidth]{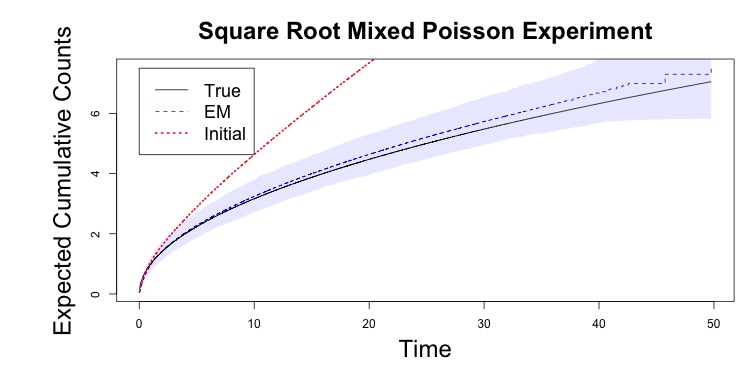}\includegraphics[width=0.5\textwidth]{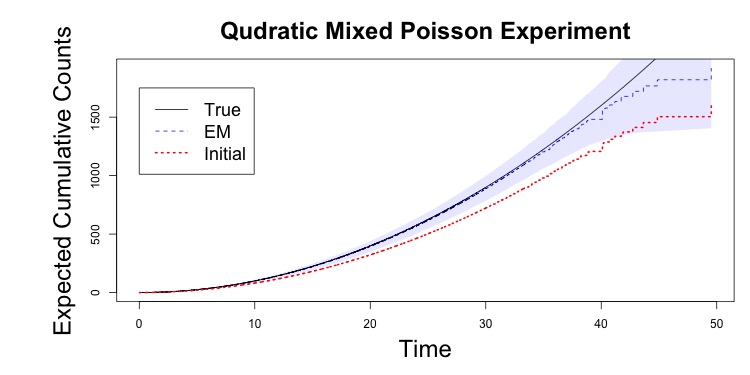}
    \caption{Mean+95\% Monte Carlo CIs of 1000 runs of Mixed Poisson processes, which violate the Poisson process assumption. Synthetic datasets: 20\% missingness using AEE wrapped with our method.  Black/solid is the true mean function.  Red/dotted is initialization with missing data set to Poisson$(1)$ values.  Blue/dashed is our EM algorithm initialized at red. (a) $\Lambda^*(u)=\sqrt{u}$ a square root mean function (b) $\Lambda^*(u)=u^2$ a quadratic mean function}
    \label{fig:square-root}
\end{figure}

\section{Further Analysis of Bladder Tumor Dataset}\label{sec:bladder-tumor-appendix}

\subsection{Comment on Mean and Confidence Intervals}

In order to form the mean and marginal confidence intervals, we can note that the observation times are all discrete valued at times (months) 1 to 50. Thus we can simply take mean and marginal confidence intervals at those points. To do so, for each bootstrap replicate, we add a small amount of noise $\pm 1e-6$ to the observed time points for identifiability purposes in order to train, and then take mean and marginal confidence intervals at the described points.

\subsection{Varying Missingness Probabilities}
In this experiment we vary the missingness probability with $\epsilon=0.1,0.2,0.3,0.4$.  We do so for each of the five following methods: the non-parametric step function maximum pseudo-likelihood (NPMPLE) of \cite{wellner2000two}, the smoothed maximum pseudo-likelihood (MPLs) and and maximum likelihood (MLs) estimators of \cite{lu2007estimation}, and the step function solution to the augmented estimating equation (AEE) and the informative censoring (AEEX) version methods of \cite{wang2013augmented}.  We show that bias is low but increasing as a function of the missingness probability for the five methods. Figure \ref{fig:bladder-varied-missingness} shows results.  The bias is very low in all cases. Further, in all cases it is much lower than the initialization with corrupted data shown in Figure \ref{fig:results}a.

\begin{figure}
    \centering
    \includegraphics[width=0.5\textwidth]{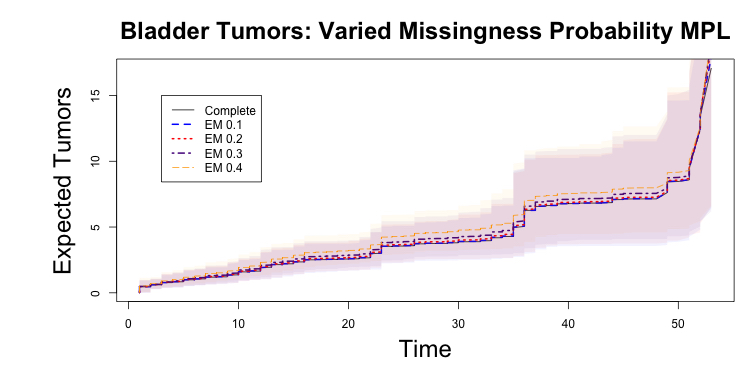}\includegraphics[width=0.5\textwidth]{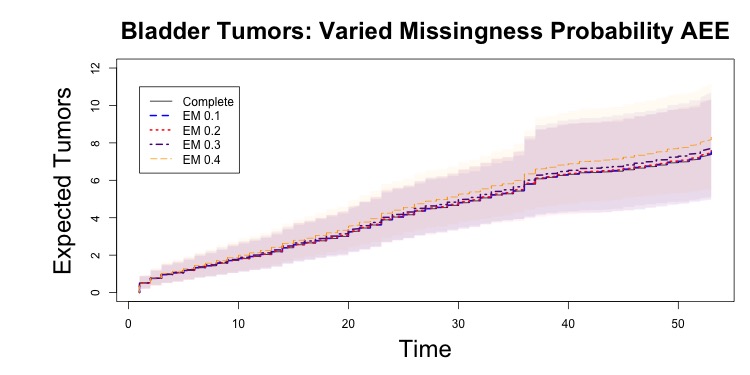}\\\includegraphics[width=0.5\textwidth]{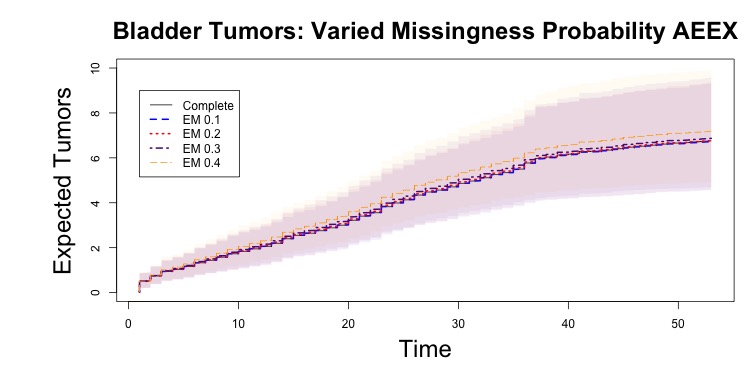}\includegraphics[width=0.5\textwidth]{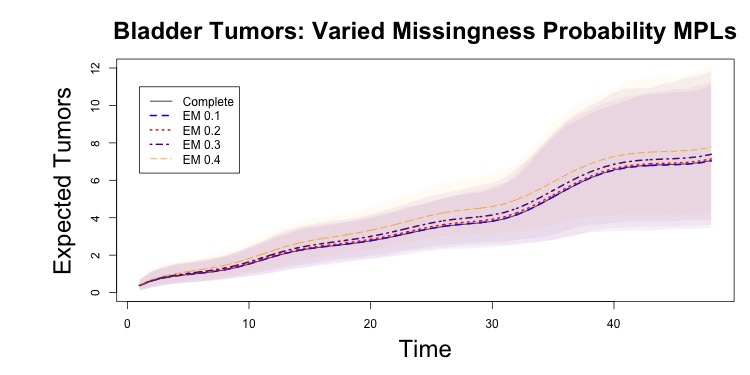}\\
    \includegraphics[width=0.5\textwidth]{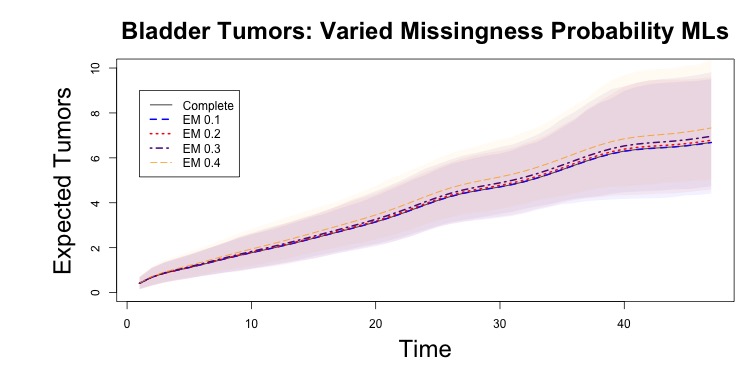}
    \caption{results for mean and 95\% CIs for 1000 bootstrap replicates for bladder tumor dataset, with the missingness probability $\epsilon$ set to $\epsilon=0.1,0.2,0.3,0.4$.  We see that bias is very low in all settings, although it increases with the missingness probability. The methods wrapped with our method are (a) maximum pseudo-likelihood (MPL) of \cite{wellner2000two} (b) augmented estimating equations (AEE) of \cite{wang2013augmented} (c) augmented estimating equations with informative censoring (AEEX) of \cite{wang2013augmented} (d) maximum pseudo-likelihood splines (MPLs) of \cite{lu2007estimation} note that for this we need to reduce the time axis slightly as spline models cannot interpolate far past the region where they have values in learning (e) maximum likelihood splines (MLs) of \cite{lu2007estimation}, same issue.}
    \label{fig:bladder-varied-missingness}
\end{figure}

\subsection{Varying Initialization}

In this experiment we vary the initialization. We note that the mean count across all observations is $0.44$. We investigate ${\sf Poisson}(1)$ to ${\sf Poisson}(4)$ initializations under both $\epsilon=0.2$ and $\epsilon=0.4$. We use AEE in all cases. Figure \ref{fig:varied-initialization} shows results. We see that further initializations increase bias, but only slightly for $\epsilon=0.2$ and more so for $\epsilon=0.4$. This is not surprising as our theory requires good initialization and sufficiently low missingness probability.

\begin{figure}
    \centering
    \includegraphics[width=0.49\textwidth]{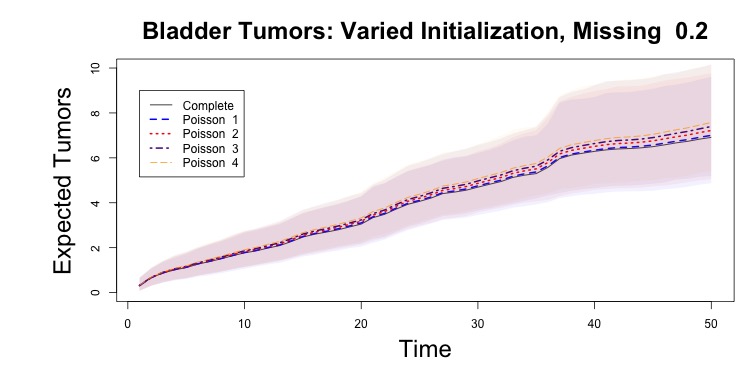}
    \includegraphics[width=0.49\textwidth]{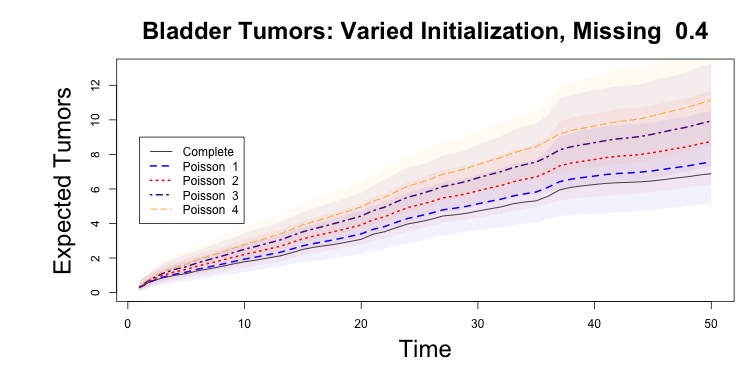}
    \caption{Here we vary the initialization by initializing increasingly far from the mean observed under complete data. The true sample mean of all intervals is $0.44$. We initialize to Poisson random variables with means $1$ to $4$. We see that initializing further from the truth increases the bias, and that it is worse for the higher missingness probability of $0.4$. Our theory requires good initialization and the missingness probability to be sufficiently low, so this is not surprising.}
    \label{fig:varied-initialization}
\end{figure}

\subsection{Heterogeneity in Missingness Probabilities}

Even if the MCAR assumption does hold, the missingness probability may vary between subjects. In this experiment, we let $\epsilon|X=\epsilon_{\textrm{mean}}X$ where $X\sim \textrm{uniform}(0,2)$ and $\epsilon_{\textrm{mean}}=0.2,0.4$. This can capture between subject heterogeneity in missingness. Within each subject, we compare initialization of ${\sf Poisson}(1)$ to ${\sf Poisson}(4)$. Figure \ref{fig:missingness-heterogeneity} shows results using the AEE method. They look very similar to results without heterogeneity in the missingness probabilities between subjects.

\begin{figure}
    \centering
    \includegraphics[width=0.49\textwidth]{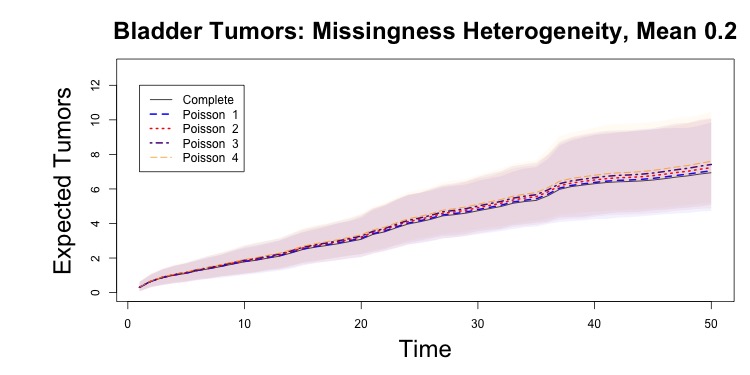}
    \includegraphics[width=0.49\textwidth]{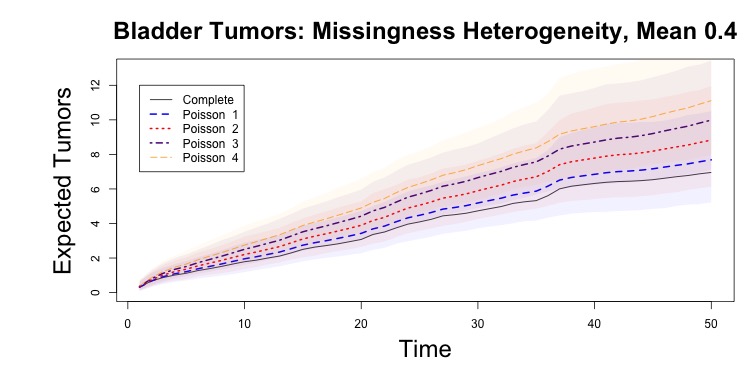}
    \caption{Here we let the missingness probability vary per participant while still being MCAR by multiplying the mean missingness by some value drawn from $\textrm{uniform}[0,2]$. We again initialize to ${\sf Poisson}(1)$ to ${\sf Poisson}(4)$ for a) mean missingness $0.2$ and b) mean missingness $0.4$.}
    \label{fig:missingness-heterogeneity}
\end{figure}

\subsection{Missing at Random}
Here we investigate the recovery under missing at random (MAR). For MAR, we note that 83.7\% of the counts are $0$: that is, there are no new tumors. We then set the missingness probability for a specific observation for a participant to $\epsilon_j=0.10$ if $\Delta N_{j-1}=0$ and otherwise $\epsilon_j=0.3$. Figure \ref{fig:MAR-MNAR} shows the results. Recovery is very good.

\begin{figure}
    \centering
    \includegraphics[width=0.49\textwidth]{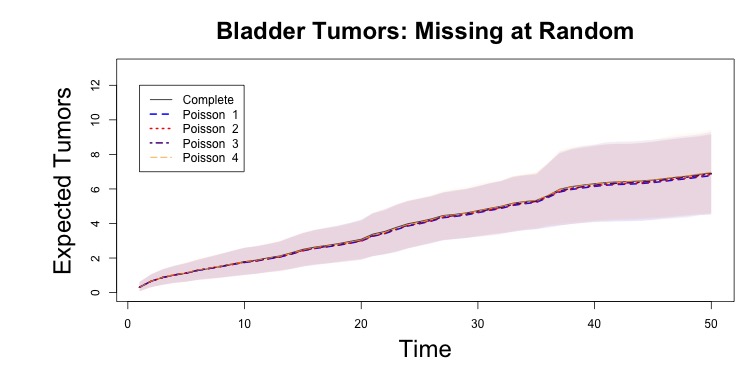}
    \caption{Missing at random (MAR). The time-varying missingness probability is $\epsilon_j=0.1$ if the previous $\Delta N_{j-1}=0$ (approximately 84\% of observations) and $\epsilon_j=0.3$ otherwise.}
    \label{fig:MAR-MNAR}
\end{figure}

\section{Further Exploratory Analysis of the Smoking Cessation Study and Dataset}\label{sec:smoking-cessation-appendix}

In this section we plot several histograms and boxplots to get insight into behavior variability in the study. We plot cigarettes since the last assessment and number of days between EMAs, where for both we aggregate both between and within subjects. We also plot the percentage of long intervals between subjects. That is, we take the percentage for each participant and then take the boxplots and histograms with the summary statistic for each participant as a data point.

Figure \ref{fig:cigs_since_last} shows the plots for number of cigarettes since the last assessment. We see that most of the time they don't smoke any cigarettes. The histogram looks like a geometric distribution for the number of cigarettes. The mean number of cigarettes is 1.76, the median is 0, and the 25 and 75 percentiles are 0 and $2$, respectively. Thus, in at least half of the intervals they don't smoke, and in 75\% percent of them they smoke at most two cigarettes, but in some cases they smoke some huge number: in one case someone smoked over $40$ cigarettes. That said, the number of cigarettes also will depend on the length of time elapsed since the last assessment, and the intensity function is time-varying: particularly, we noted in Figure \ref{fig:results}b) that they tend to smoke more frequently in the pre-quit period.

Figure \ref{fig:days_between_EMAs} plots the days between EMAs. Here we see that most observations are under one day, but there are a substantial number of outliers. In a number of cases the time between EMAs is over two days, and in several cases it is over six days. The mean time between observations is 0.34 days or approximately eight hours, the median is 0.17 or approximately four hours, and the 25 and 75 percentile are approximately two and 11 hours, respectively. Summarizing, most inter-EMA durations are under one day, but a few are very long.

Finally, also relevant is whether there is heterogeneity in the proportion of long intervals between subjects. Figure \ref{fig:long_intervals_between_subjects} investigates this. We see that participants vary between having no unreliable intervals and in one case having all unreliable intervals. The mean is 11\%, the median is 6\%, and the 25 and 75 percentile are 1.8\% and 11.4\%, respectively. Thus there is substantial heterogeneity. While our theory does not currently account for this and we leave it to future work, the experiments from the previous section where each subject had a different missingness probability (Figure \ref{fig:missingness-heterogeneity}) suggests that our model can handle this issue well.

\begin{figure}
    \centering
    \includegraphics[width=0.49\textwidth]{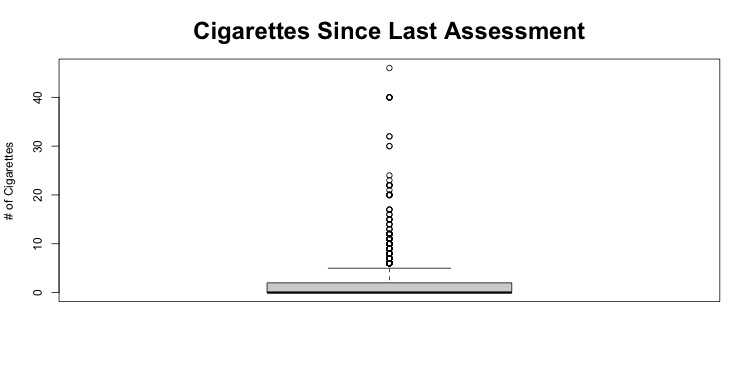}
    \includegraphics[width=0.49\textwidth]{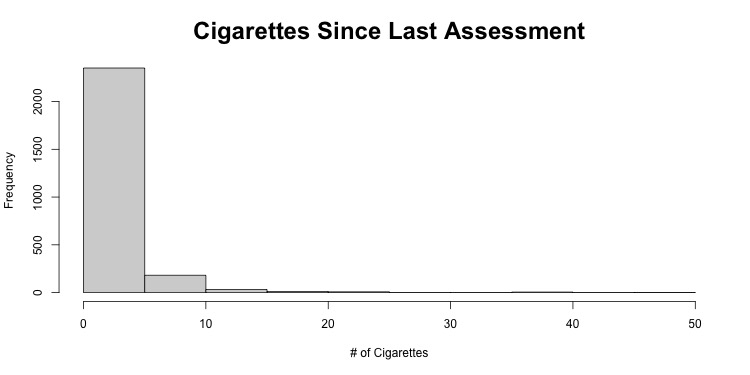}
    \caption{a) boxplot b) histogram for number of cigarettes since the last assessment. We see that the median number of cigarettes is 0, and that there are many outliers.}
    \label{fig:cigs_since_last}
\end{figure}

\begin{figure}
    \centering
    \includegraphics[width=0.49\textwidth]{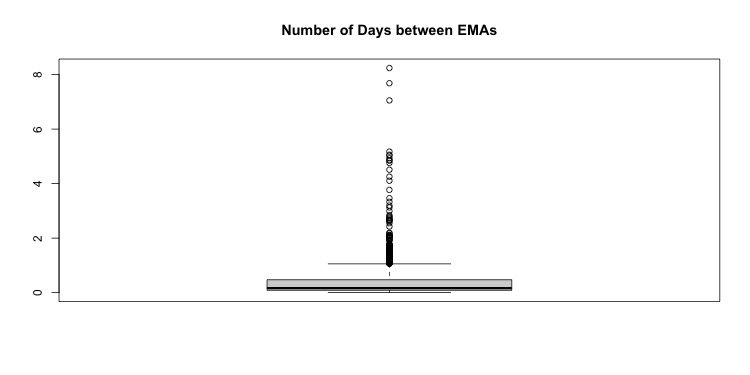}
    \includegraphics[width=0.49\textwidth]{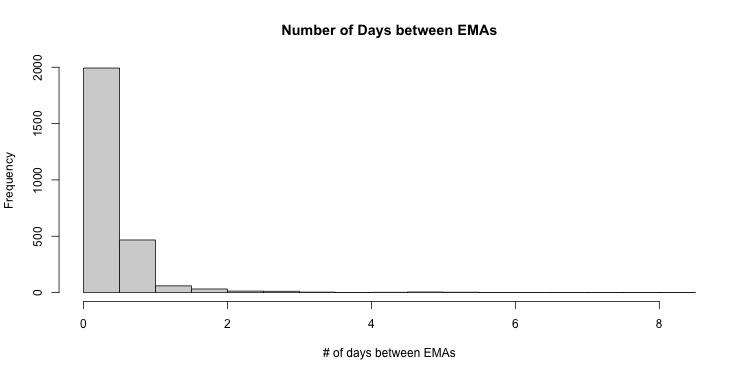}
    \caption{a) boxplot b) histogram for number of days between EMAs. Recall that we treat intervals over one day as missing/unreliable. From the boxplot we see that the mean and 75\% percentile are under one day, but that there are a substantial number of observations over one day. The histogram suggests a similar finding.}
    \label{fig:days_between_EMAs}
\end{figure}

\begin{figure}
    \centering
    \includegraphics[width=0.49\textwidth]{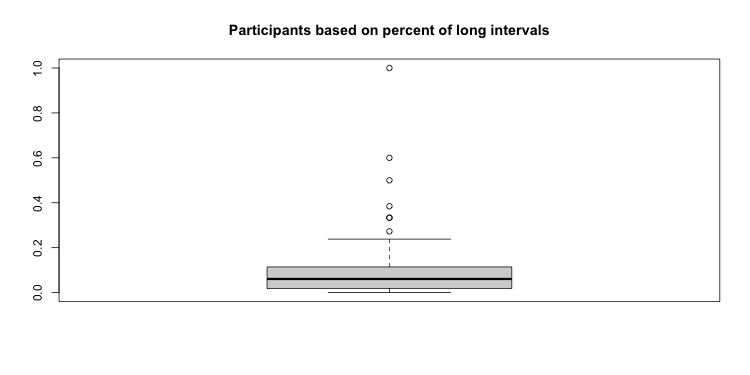}
    \includegraphics[width=0.49\textwidth]{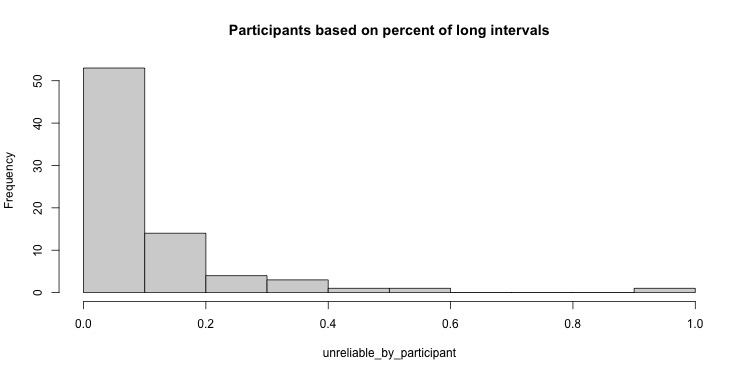}
    \caption{a) boxplot b) histogram for percent of long intervals between subjects. We see that most participants have a small proportion of long intervals (the median is 6.1\%), while a few have a much larger proportion of long intervals.}
    \label{fig:long_intervals_between_subjects}
\end{figure}


\end{document}